\documentclass[draftcls,onecolumn]{IEEEtran}

\usepackage[utf8]{inputenc}
\usepackage[T1]{fontenc}
\usepackage{url}
\usepackage{ifthen}
\usepackage{cite}
\usepackage[cmex10]{amsmath} 
\usepackage{graphicx}
\usepackage{amsmath, amsthm, amssymb, amsfonts}
\usepackage{tabularx}
\usepackage[linesnumbered,ruled]{algorithm2e}
\usepackage{algpseudocode}
\usepackage{subfigure}
\usepackage{url}
\usepackage[monochrome]{xcolor}
\usepackage{multirow}
\usepackage{colortbl}

\newtheorem{lemma}{Lemma}

\newtheorem{theorem}{Theorem}
\newtheorem{remark}{Remark}

\begin{document}
\title{The Construction of Near-Optimal Universal Coding of Integers}

\author{Wei Yan,
        and Yunghsiang S. Han,~\IEEEmembership{Fellow,~IEEE}% <-this % stops a space
\thanks{
W. Yan is with the College of Electronic Engineering, National University of Defense Technology, Hefei 230037, China, and also with the College of Computer Science and Technology, National University of Defense Technology, Changsha, 
410073, China. Email: yan.wei2023@nudt.edu.cn.}
\thanks{Y. S. Han is affiliated with the Shenzhen Institute for Advanced Study, University of Electronic Science and Technology of China, Shenzhen 518000, China. Email: yunghsiangh@gmail.com.}
}
\maketitle
\begin{abstract}
The Universal Coding of Integers~(UCI) is suitable for discrete memoryless sources with unknown probability distributions and infinitely countable alphabet sizes.
A UCI is a class of prefix codes for which the ratio of the average codeword length to $\max\{1,H(P)\}$ is within a constant expansion factor \textcolor{red}{$C_{\mathcal{C}}$} for any decreasing probability distribution $P$, where $H(P)$ is the entropy of $P$.
For any UCI code $\mathcal{C}$, \emph{the minimum expansion factor} \textcolor{red}{$C_{\mathcal{C}}^{*}$} is defined to represent the infimum of the set of extension factors of $\mathcal{C}$. Each $\mathcal{C}$ has a unique corresponding \textcolor{red}{$C_{\mathcal{C}}^{*}$}, and the smaller \textcolor{red}{$C_{\mathcal{C}}^{*}$} is, the better the compression performance of $\mathcal{C}$ is.
The class of UCIs $\mathcal{C}$ (or a family $\{\mathcal{C}_i\}_{i=1}^{\infty}$) that achieves the smallest \textcolor{red}{$C_{\mathcal{C}}^{*}$} is defined as the \emph{optimal UCI}.
The best current result is that the range of $C_{\mathcal{C}}^{*}$ for the optimal UCI is $2\leq C_{\mathcal{C}}^{*}\leq 2.5$.
In this paper, we prove a tighter probability inequality for decreasing distributions, which serves as a new tool for studying 
the properties of UCIs. On the basis of
this inequality, we prove that there exists a class of near-optimal UCIs, called the $\nu$ code, achieving \textcolor{red}{$C_\nu=2.0386$}. This narrows the range of the minimum expansion factor for the optimal UCI to $2\leq C_{\mathcal{C}}^{*}\leq 2.0386$. We show that the $\nu$ code is currently optimal in terms of the minimum expansion factor.
In addition, we propose a new proof showing that the minimum expansion factor of the optimal UCI is lower bounded by $2$.
\end{abstract}
\IEEEpeerreviewmaketitle

\section{Introduction}\label{sec_intro}
In their pioneering work~\cite{SH481,SH482} Shannon proposed a model for general communication systems, defined Shannon's entropy, initiated the study of source coding, and constructed the first variable-length code~(VLC), which was termed the Shannon--Fano coding.
Source coding improves the effectiveness of digital information storage, processing, and transmission methods by reducing redundancy.
VLCs form an essential source coding class that encodes individual source symbols into variable-length codewords.
The idea behind constructing a VLC is to map source symbols with higher occurrence probabilities into shorter codewords, thereby achieving a shorter average codeword length.
This idea emerged as early as the nineteenth century, when Morse developed Morse code.
When the probability distribution of discrete memoryless sources~(DMSs) is known, the Huffman code~\cite{H52} is the prefix code with the shortest average codeword length, where the prefix code is a special type of VLC that has any codeword that is not a prefix of any other codeword. The arithmetic code~\cite{P76,AC79} is another well-known source code; however, it is a string coding method that falls outside the scope of this work.

However, in practical applications, it is sometimes difficult to know precisely about the statistical properties of the source or even to know anything.
When the probability distribution of DMSs is unknown, the Universal Coding of Integers (UCI) is the best-known class of VLCs.
A systematic study of UCIs was initiated by Elias~\cite{Elias75} in 1975, and UCIs were determined to apply to DMSs of the countably infinite alphabet, with the assumption of no prior knowledge.
Today, UCIs are widely used in various areas, including data deduplication schemes~\cite{N19,LF22}, stream processing engines~\cite{Zhang23,Zhang24}, biological sequencing with data compression~\cite{DNA10,DNA13}, quantized stochastic gradient descent~\cite{NIPS}, and evolving secret sharing~\cite{YLH23,Cheng25}.

Elias~\cite{Elias75} considered the following universal coding problem.
The source $S=(\mathcal{A},P)$ is discrete and memoryless, where the alphabet $\mathcal{A}\triangleq\mathbb{N}=\{1,2,\cdots\}$ is countably infinite and the probability distribution $P$ is decreasing (i.e., \textcolor{red}{$P(a)\geq P(a+1)$} for all $a\in \mathcal{A}$ and $\sum_{k=1}^{\infty}P(k)=1$).\footnote{In this paper, we assume that any probability distribution $P$ is decreasing.}
Let $H(P)\triangleq-\sum_{k=1}^{\infty}P(k)\log_{2}P(k)$ be the Shannon entropy of $P$.
Let $\mathcal{C}$ be a prefix code for the DMS $S=(\mathcal{A},P)$; that is, $\mathcal{C}$ maps the positive integers onto binary codewords and satisfies the prefix property.
Let $L_{\mathcal{C}}(\cdot)$ be the length function such that $L_{\mathcal{C}}(a)=|\mathcal{C}(a)|$ for all $a\in \mathcal{A}$,
where $\mathcal{C}(a)$ denotes the $a$-th codeword of $\mathcal{C}$.
Let $A_{P}(L_{\mathcal{C}})\triangleq\sum_{k=1}^{\infty}P(k)L_{\mathcal{C}}(k)$ denote the average codeword length of $\mathcal{C}$ under the probability distribution $P$.
The prefix code $\mathcal{C}$ is said to be \emph{universal} if a constant $C_{\mathcal{C}}$ exists such that
\begin{equation}\label{eq1}
\frac{A_{P}(L_{\mathcal{C}})}{\max\{1,H(P)\}}\leq C_{\mathcal{C}}
\end{equation}
for all $P$ with $0<H(P)<\infty$.
The constant $C_{\mathcal{C}}$ is called the \emph{expansion factor} of UCI $\mathcal{C}$, and we let $S_{\mathcal{C}}$ denote the set consisting of all expansion factors of UCI $\mathcal{C}$.

Since Elias's pioneering work~\cite{Elias75}, many UCIs have been proposed, and they can be categorized into two main groups~\cite{C1990}:
\emph{the message length scheme} and \emph{the flag scheme}.
\begin{enumerate}
\item In the message length scheme, the process of encoding an integer $n\in\mathbb{N}$ can be divided into two parts,
where the suffix part encodes $n$ with a length of $m$ bits and the prefix part encodes $m$.
Note that the prefix part can be further encoded into two parts, and so forth.
The $\gamma$, $\delta$, and $\omega$ codes proposed by Elias~\cite{Elias75} are typical UCIs within the message length scheme.
  A UCI $\mathcal{C}$ in this scheme mainly involves keeping the codeword length $L_{\mathcal{C}}(n)$ as small as possible for sufficiently large values of $n\in\mathbb{N}$.
 When the length of the codeword at infinity is as small as possible, the length of the codeword at smaller integers is relatively large; thus, this research focuses on the case where the Shannon entropy is large. UCIs for such schemes were proposed in~\cite{Elias75,L68,ER78,S80,Y00}.
\item In the flag scheme, a specific string is used as a flag to identify the end of a codeword.
 For example, the Fibonacci code~\cite{AF87} is a typical UCI within the flag pattern scheme, and $11$ is its flag.
 The UCIs in the message length scheme can catastrophically propagate errors into subsequent codewords during the decoding process.
However, the flag scheme used by UCIs has a decoding synchronization property that prevents the infinite propagation of errors.
   UCIs for such schemes were proposed in~\cite{L81,AF87,W88,YO91,AS17}.
\end{enumerate}

Recently, Allison~\emph{et al.}~\cite{DCC21} analysed the main properties of Wallace tree codes and proved that they are UCIs. 
\cite{YL21,YL22} defined the \emph{minimum expansion factor} $C_{\mathcal{C}}^{*}\triangleq \inf S_{\mathcal{C}}$.
For any UCI $\mathcal{C}$, its minimum expansion factor $C_{\mathcal{C}}^{*}$ is unique.
The minimum expansion factor characterizes the compression performance of a UCI.
If a class of UCIs $\mathcal{C}$ or a family of UCIs $\{\mathcal{C}_i\}_{i=1}^{\infty}$ possesses the smallest minimum expansion factor,\footnote{The minimum expansion factor corresponding to a family of UCIs $\{\mathcal{C}_i\}_{i=1}^{\infty}$ is $\lim\limits_{i\rightarrow\infty}C_{\mathcal{C}_i}^*$.}
then $\mathcal{C}$ or $\{\mathcal{C}_i\}_{i=1}^{\infty}$ is said to be the \emph{optimal UCI}.
Yan \emph{et al.} proved that the range of $C_{\mathcal{C}}^{*}$ for the optimal UCI is $2\leq C_{\mathcal{C}}^{*}\leq 2.5$.
In addition, a generalized universal coding of integers related to UCIs has been proposed and investigated~\cite{ITW,YH24,Yan25}.

In this paper, we further narrow the range of $C_{\mathcal{C}}^{*}$ for the optimal UCI.
By establishing a tighter probability inequality for decreasing distributions and developing an improved research approach, we prove that the constructed $\nu$ code is currently optimal in terms of the minimum expansion factor.
The contributions of this paper are listed below.
\begin{enumerate}
\item Inspired by minimax regret data compression, we propose a new proof showing that the minimum expansion factor of the optimal UCI is lower bounded by $2$.
\item \textcolor{blue}{We prove a tighter probability inequality for decreasing distributions, which serves as a new tool for studying the properties of UCIs. On the basis of this inequality, we improve the research approach for determining the upper bound of the minimum expansion factor for the optimal UCI.}
\item We propose a new class of UCIs, termed the $\nu$ code, and prove that the $\nu$ code is currently optimal in terms of the minimum expansion factor with $C_{\mathcal{\nu}}=2.0386$. This reduces the upper bound of the minimum expansion factor for the optimal UCI from $2.5$ to $2.0386$.
\end{enumerate}

The paper is structured as follows.
Section \ref{pre} provides the relevant conclusions regarding UCIs and Abel's formula.
In Section \ref{new}, a new proof stating that the minimum expansion factor of the optimal UCI is lower bounded by $2$ is proposed.
Section \ref{sec_idea} presents a tighter probability inequality for decreasing distributions and introduces an improved research approach.
Section \ref{newcode} involves the construction of a new class of UCIs, termed the $\nu$ code.
In Section \ref{proof}, the upper bound of $C_{\nu}^{*}$ is proven.
Section \ref{sec_dis} provides the lower bound of $C_{\nu}^{*}$ and briefly compares the range of $C_{\nu}^{*}$ with that of $C_{\mathcal{C}}^{*}$ for the classic UCIs.
Section \ref{sec_con} concludes this work.

\section{Preliminaries}\label{pre}
In this section, we present the Elias $\delta$ code, recent advances regarding minimum expansion factors for UCIs and Abel's formula.
Table~\ref{tab3} presents the main notations used in this paper.

\begin{table}[t]
\centering
\color{red}
\caption{Main notations used in this paper}\label{tab3}
\begin{tabular}{c|l|c}
\hline \hline
Notation  & Description  & Definition \\
\hline \hline
$\mathcal{A}$    & The countably infinite alphabet  &   Section \ref{sec_intro}                 \\
\hline
$P$   &     The decreasing probability distribution &  Section \ref{sec_intro}     \\
\hline
$S=(\mathcal{A},P)$ & The discrete memoryless source with an alphabet $\mathcal{A}$ and a probability distribution $P$ &  Section \ref{sec_intro}          \\
\hline
$H(P)$   &    The Shannon entropy of $P$  & Section \ref{sec_intro}      \\
\hline
$\mathcal{C}$   &   The prefix code for $S=(\mathcal{A},P)$  &  Section \ref{sec_intro}         \\
\hline
$\mathcal{C}(a)$  & The $a$-th codeword of $\mathcal{C}$ for $a\in\mathcal{A}$ & Section \ref{sec_intro}                          \\
\hline
$L_{\mathcal{C}}(a)$  &   The length of $\mathcal{C}(a)$   & Section \ref{sec_intro}     \\
\hline
$A_P(L_{\mathcal{C}})$   & The average codeword length of $\mathcal{C}$ under a probability distribution $P$ & Section \ref{sec_intro}   \\
\hline
$C_{\mathcal{C}}$  & The expansion factor of UCI $\mathcal{C}$    &   Eq.~\eqref{eq1}         \\
\hline
$S_{\mathcal{C}}$ &  The set consisting of all expansion factors of UCI $\mathcal{C}$ & Section \ref{sec_intro}  \\
\hline
$C_{\mathcal{C}}^*$  & The minimum expansion factor of UCI $\mathcal{C}$    &  Section \ref{sec_intro}         \\
\hline
$\{0,1\}^{*}$   & The set consisting of all binary strings of finite length   &    Section \ref{pre}          \\
\hline
$|\xi|$     &    The length of string $\xi$             &   Section \ref{pre}    \\
\hline
$\mathcal{P}_\mathcal{A}$  & The set of all decreasing probability distributions over $\mathcal{A}$ with finite Shannon entropy   &  Section \ref{new}         \\
\hline
$\mathcal{UCI}$            &  The set of all UCIs    &  Section \ref{new}         \\
\hline
$C^*$                 &   The minimum expansion factor of the optimal UCI    &   Section \ref{new}      \\
\hline
$h(p)$               &   The binary entropy of $(p,1-p)$         &   Section \ref{new}      \\
\hline
$A_n$                  &   $A_n=\log_{2}P(n)-\log_{2}(1-P(1))+\log_{2}n$    &    Eq.~\eqref{eq40}                   \\
\hline
$C_n$                  &   $C_n=(1+\frac{1}{n})(1-\frac{1}{n})^{n-1}$    &    Section \ref{sec_idea}                   \\
\hline
$g_{(c_1,c_2)}\left(a,P(1)\right)$  &          $g_{(c_1,c_2)}\left(a,P(1)\right)=\frac{c_1}{c_2-\log_{2}\big(1-P(1)\big)}\Big(\log_{2}a-\log_{2}\big(1-P(1)\big)\Big)$  &
 Eq.~\eqref{eq45}     \\
\hline
$h_{(c_1,c_2)}\left(t,x\right)$ & $h_{(c_1,c_2)}\left(t,x\right)=\left(\frac{c_1}{c_2+x}-1\right)t+\frac{c_1x}{c_2+x}-1$  &
 Eq.~\eqref{eq46}     \\
\hline
  $D_{(c_1,c_2)}(P(1))$     &    $D_{(c_1,c_2)}(P(1))= \frac{c_1}{c_2-\log_{2}\big(1-P(1)\big)}$ &          \\
 $J_{(c_1,c_2)}(P(1),L)$    &  $J_{(c_1,c_2)}(P(1),L)\!=\! D_{(c_1,c_2)}(P(1))\log_{2}C_L\!+\!P(1)\Big[1\!+\!D_{(c_1,c_2)}(P(1))\big(\log_{2}P(1)\!-\!\log_{2}C_{L}\big)\Big]$  &   Eq.~\eqref{eq52}     \\
$R_{(c_1,c_2)}(P(1),L) $        & $R_{(c_1,c_2)}(P(1),L)=3\!+\!D_{(c_1,c_2)}(P(1))\left(\log_{2}\big(1\!-\!P(1)\big)+\sum_{n=2}^{L-1}\log_{2}C_{n} \!-\!(L\!-\!1)\log_{2}C_{L}\right)$   &      \\
$Q_{(c_1,c_2)}(P(1),P(2),L)$ &    $Q_{(c_1,c_2)}(P(1),P(2),L)=J_{(c_1,c_2)}(P(1),L)+P(2)R_{(c_1,c_2)}(P(1),L)$ &        \\
\hline \hline
\end{tabular}
\end{table}

\subsection{The Elias $\delta$ code}
Let $\{0,1\}^{*}$ denote all binary strings of finite length and $|\xi|$ denote the length of string $\xi$.
\textcolor{red}{Let an $\alpha$ code be a unary encoding~\cite{Elias75,book07}.
The unary code of $a\in\mathcal{A}$ consists of $a-1$ zeros followed by a single one.}
For example, $\alpha(1)=1$, and $\alpha(4)=0001$. Let a $\beta$ code be a binary representation.
For example, $\beta(4)=100$, and $\beta(11)=1011$. Let $[\beta(a)]$ denote the removal of the most significant bit $1$ of $\beta(a)$.
In particular, $[\beta(1)]$ is a null string.
Then, the codeword length is as follows.
\begin{equation*}
\begin{aligned}
			|\alpha(a)|&=a,   \\
			|\beta(a)|&=1+\lfloor \log_{2} a\rfloor,        \\
			|[\beta(a)]|&=\lfloor \log_{2} a\rfloor
\end{aligned}
\end{equation*}
for all $a\in\mathcal{A}$. The Elias $\gamma$ code~\cite{Elias75}: $\mathcal{A}\rightarrow \{0,1\}^{*}$ can be succinctly represented as
\[
	\gamma(a)=\textcolor{red}{\alpha}(|\beta(a)|)[\beta(a)]
\]
for all $a\in\mathcal{A}$.
The Elias $\delta$ code~\cite{Elias75}: $\mathcal{A}\rightarrow \{0,1\}^{*}$ can be succinctly represented as
\[
	\delta(a)=\gamma(|\beta(a)|)[\beta(a)]
\]
for all $a\in\mathcal{A}$. Afterward, we obtain
\begin{equation*}
\begin{aligned}
			L_\delta(a)&=\big|\gamma(|\beta(a)|)\big|+\big|[\beta(a)]\big|    \\
			           &=\big|\gamma(1+\lfloor \log_{2} a\rfloor)\big|+\lfloor \log_{2} a\rfloor          \\
                       &=\big|\alpha(|\beta(1+\lfloor \log_{2} a\rfloor)|)\big|+\big|[\beta(1+\lfloor \log_{2} a\rfloor)]\big|+\lfloor \log_{2} a\rfloor   \\
                       &=\big|\beta(1+\lfloor \log_{2} a\rfloor)\big|+\lfloor \log_{2}(1+\lfloor \log_{2} a\rfloor)\rfloor+\lfloor \log_{2} a\rfloor   \\
			           &=1+\lfloor\log_{2}a\rfloor+2\lfloor\log_{2}(1+\lfloor\log_{2}a\rfloor)\rfloor
\end{aligned}
\end{equation*}
for all $a\in\mathcal{A}$.
Table~\ref{tab1} lists the first 16 codewords of the Elias $\gamma$ and Elias $\delta$ codes.
Additional details related to the Elias code can be found at~\cite{Elias75}.
\begin{table*}[!t]
\centering
\caption{The first 16 codewords of the Elias $\gamma$ and Elias $\delta$ codes}\label{tab1}
\begin{tabular}{c|c|c}
\hline \hline
$n$  &  Elias $\gamma$ code  & Elias $\delta$ code    \\
\hline \hline
$1$   &      1        &  1              \\
$2$   &     01 0      &  010 0     \\
$3$   &     01 1      &  010 1      \\
$4$   &     001 00    &  011 00    \\
$5$   &     001 01    &  011 01     \\
$6$   &     001 10    &  011 10     \\
$7$   &     001 11    &  011 11     \\
$8$   &     0001 000  &  00100 000     \\
$9$   &     0001 001  &  00100 001    \\
$10$  &     0001 010  &  00100 010  \\
$11$   &    0001 011  &  00100 011     \\
$12$   &    0001 100  &  00100 100    \\
$13$   &    0001 101  &  00100 101   \\
$14$   &    0001 110  &  00100 110     \\
$15$   &    0001 111  &  00100 111    \\
$16$   &    00001 0000 & 00101 0000     \\
\hline \hline
\end{tabular}
\end{table*}

\subsection{The minimum expansion factor of UCIs}
The minimum expansion factor $C_{\mathcal{C}}^{*}$ of a UCI $\mathcal{C}$ characterizes the compression performance of $\mathcal{C}$; that is,
whatever the probability distribution $P$ of the source is, the average codeword length $A_{P}(L_{\mathcal{C}})$ is less than or equal to $C_{\mathcal{C}}^{*}$ times $\max\{1,H(P)\}$.
Elias~\cite{Elias75} proved that the minimum expansion factor $C_{\mathcal{C}}^{*}$ has a trivial lower bound of $1$, and Yan \emph{et al.}~\cite{YL21} first introduced a nontrivial lower bound for $C_{\mathcal{C}}^{*}$, as described below.
\begin{theorem}~\cite{YL21}\label{thm1}
The minimum expansion factor $C_{\mathcal{C}}^{*}\geq2$ for any UCI $\mathcal{C}$.
\end{theorem}
Theorem~\ref{thm1} was proven in~\cite{YL21} by the following lemma.
\begin{lemma}~\cite{YL21}
Suppose that $\mathcal{C}$ is any given UCI and that $K_1$ and $K_2$ are two positive constants.
\begin{enumerate}
\item If the length function $L_\mathcal{C}(\cdot)$ satisfies $L_\mathcal{C}(a)\geq K_1+K_2\lfloor\log_{2}a\rfloor$ for all $a\in\mathcal{A}$, then $C_{\mathcal{C}}^{*}\geq K_1+K_2$.
\item If $L_\mathcal{C}(\cdot)$ satisfies $L_\mathcal{C}(a)\leq L_\mathcal{C}(a+1)$ for all $a\in\mathcal{A}$, then $L_\mathcal{C}(a)\geq 1+\lfloor\log_{2}a\rfloor$ for all $a\in\mathcal{A}$.
\end{enumerate}
\end{lemma}

Table \ref{tab2}~\cite{YL22} summarizes the latest research results that have been published regarding the minimum expansion factors $C_{\mathcal{C}}^{*}$ of some UCIs.
In Table \ref{tab2}, the $\iota$ code~\cite{YL22} constructed by Yan \emph{et al.} can reach an expansion factor $C_{\iota}=2.5$.
Thus, the range of $C_{\mathcal{C}}^{*}$ for the optimal UCI is $2\leq C_{\mathcal{C}}^{*}\leq 2.5$.
\begin{table}[t]
\centering
\color{violet}
\caption{The minimum expansion factors $C_{\mathcal{C}}^{*}$ of some UCIs}\label{tab2}
\begin{tabular}{|c|c|}
\hline
Code  & The range of $C_{\mathcal{C}}^{*}$   \\
\hline
$\gamma$ code~\cite{Elias75}       &     $C_{\gamma}^{*}=\mathbf{3}$                  \\
$\delta$ code~\cite{Elias75}   &    $2.5 \leq C_{\delta}^{*}\leq \mathbf{2.75} $      \\
$\omega$ code~\cite{Elias75}   &    $2.1 < C_{\omega}^{*}\leq \mathbf{3} $    \\
$CE(k)$ code~\cite{A1993}         &     $2+k\leq C_{CE(k)}^{*}\leq \mathbf{2+k+\frac{1}{2^{k}}} $           \\
$\eta$ code~\cite{YL21}    &   $ 2.5 \leq C_{\eta}^{*}\leq \mathbf{\frac{8}{3}} $    \\
$\theta$ code~\cite{YL21}    &   $ 2.5 \leq C_{\theta}^{*}\leq \mathbf{2.8}$  \\
\textbf{$\iota$ code}~\cite{YL22}   &     $C_{\iota}^{*}=\mathbf{2.5}$                \\
$\kappa$ code~\cite{YL22}   &   $ 2.5 \leq C_{\kappa}^{*}\leq \mathbf{\frac{8}{3}}$ \\
$\kappa[t]$ code~\cite{YL22}   &   $ 2.5 \leq C_{\kappa[t]}^{*}\leq \mathbf{2.5+\frac{1}{2t+2}}$   \\
\hline
\end{tabular}
\end{table}

\subsection{Abel's formula}
We introduce Abel's formula, which is used in the proof of this paper.
\begin{theorem}[Abel's formula]~\cite{Abel}
Suppose that $a_1, a_2, \ldots, a_n$ and $b_1, b_2, \ldots, b_n$ are two sets of real numbers.
Let $z_k\triangleq b_1+b_2+\cdots+b_k$ for $k=1, 2, \ldots, n$; then,
\[
  \sum_{i=1}^n a_{i}b_{i}=(a_1-a_2)z_1+(a_2-a_3)z_2+\cdots+(a_{n-1}-a_n)z_{n-1}+a_{n}z_{n}.
\]
\end{theorem}
\begin{proof}
We obtain
\begin{equation*}
\begin{aligned}
        & \quad    (a_1-a_2)z_1+(a_2-a_3)z_2+\cdots+(a_{n-1}-a_n)z_{n-1}+a_{n}z_{n}\\
        &=  a_1z_1 + a_2(z_2-z_1)+\cdots+a_n(z_n-z_{n-1})  \\
        &=  a_1b_1 + a_2b_2 +\cdots + a_nb_n.
\end{aligned}
\end{equation*}
\end{proof}
\section{The Minimum Expansion Factor of the Optimal UCI}\label{new}
This section revisits the minimal expansion factor of the optimal UCI and provides a new proof of Theorem~\ref{thm1}.
The work in this section is inspired by minimax regret data compression~\cite[chap.13]{EIT}.

Let $\mathcal{P}_\mathcal{A}$ be the set of all decreasing probability distributions over a countably infinite alphabet $\mathcal{A}$ with finite Shannon entropy values.
Let $\mathcal{UCI}$ denote the set of all UCIs.
For any UCI $\mathcal{C}$, the minimum expansion factor $C_{\mathcal{C}}^{*}$ can be equivalently defined as follows:
\begin{equation*}
\begin{aligned}
     C_{\mathcal{C}}^{*}&= \inf S_\mathcal{C} \\
                        &= \inf \left\{C_{\mathcal{C}} \Big| \frac{A_{P}(L_{\mathcal{C}})}{\max\{1,H(P)\}}\leq C_{\mathcal{C}} \mbox{ for }\forall P\in\mathcal{P}_\mathcal{A} \right\}       \\
                        &= \sup_{P\in\mathcal{P}_\mathcal{A}} \frac{A_{P}(L_{\mathcal{C}})}{\max\{1,H(P)\}} \ .
\end{aligned}
\end{equation*}
Furthermore, the minimum expansion factor of the optimal UCI can be defined as follows:
\begin{equation} \label{eq2}
\begin{aligned}
     C^{*}&= \inf_{\mathcal{C}\in\mathcal{UCI}}  C_{\mathcal{C}}^{*} \\
          &= \inf_{\mathcal{C}\in\mathcal{UCI}} \sup_{P\in\mathcal{P}_\mathcal{A}} \frac{A_{P}(L_{\mathcal{C}})}{\max\{1,H(P)\}} \ .
\end{aligned}
\end{equation}
Thus, we can precisely describe the problem of finding an optimal UCI as finding the exact value of $C^*$ and constructing a UCI $\mathcal{C}$ or a family of UCIs $\{\mathcal{C}_i\}_{i=1}^{\infty}$ achieving $C^*$.
The work of \cite{YL21,YL22} can be described as having determined that the range of $C^*$ is $2\leq C^*\leq2.5$.
Next, we provide a new proof of Theorem~\ref{thm1}.
\begin{proof}
According to~Equation \eqref{eq2}, Theorem~\ref{thm1} is equivalent to $C^{*}\geq2$.
We prove that $C^{*}\geq2$ as follows.

First, we construct some probability distribution $P_m$, where $m$ is an integer greater than $1$.
\begin{equation*}
P_m(a)=\left\{\begin{array}{lll}
1-\frac{1}{m},            &\text{if } a=1\text{ ,}\\
\frac{1}{m\cdot2^{m}},         &\text{if } a=2,3,\ldots,2^{m}+1 \text{ ,}\\
0,                          &\text{otherwise.}   \\
\end{array}\right.
\end{equation*}
Letting $\mathcal{P}_\mathcal{A}(M)\triangleq\big\{P_m \big| 2\leq m\in\mathbb{N}\big\}$, it is clear that $\mathcal{P}_\mathcal{A}(M)\subsetneqq\mathcal{P}_\mathcal{A}  $. For all $2\leq m\in\mathbb{N}$, we obtain
\begin{equation*}
\begin{aligned}
  H(P_m)&=-P(1)\log_{2}P(1)-\sum_{a=2}^{2^{m}+1}P_m(a)\log_{2}P_m(a)\\
       &=-\Big(1-\frac{1}{m}\Big)\log_{2}\Big(1-\frac{1}{m}\Big)- \frac{2^{m}}{m\cdot 2^{m}}\log_{2}\frac{1}{m\cdot 2^{m}}\\
       &=1+\textcolor{red}{h\Big(\frac{1}{m}\Big)},\\
\end{aligned}
\end{equation*}
where \textcolor{red}{$h(p)\triangleq -p\log_{2}p-(1-p)\log_{2}(1-p)$ denotes the binary entropy of $(p,1-p)$}.

For all $\mathcal{C}\in\mathcal{UCI}$ and all $\widetilde{P}\in\mathcal{P}_\mathcal{A}$,
\[
 \sup_{P\in\mathcal{P}_\mathcal{A}} \frac{A_{P}(L_{\mathcal{C}})}{\max\{1,H(P)\}}\geq \frac{A_{\widetilde{P}}(L_{\mathcal{C}})}{\max\{1,H(\widetilde{P})\}}
\]
by definition. Hence, we have that
\begin{equation} \label{eq3}
  \inf_{\mathcal{C}\in\mathcal{UCI}} \sup_{P\in\mathcal{P}_\mathcal{A}} \frac{A_{P}(L_{\mathcal{C}})}{\max\{1,H(P)\}}\geq \inf_{\mathcal{C}\in\mathcal{UCI}}  \frac{A_{\widetilde{P}}(L_{\mathcal{C}})}{\max\{1,H(\widetilde{P})\}}
\end{equation}
for all $\widetilde{P}\in\mathcal{P}_\mathcal{A}$.
Based on~Equation \eqref{eq3} and the fact that $\mathcal{P}_\mathcal{A}(M)\subsetneqq\mathcal{P}_\mathcal{A}$, we have
\begin{equation} \label{eq4}
\begin{aligned}
     C^{*}& =  \inf_{\mathcal{C}\in\mathcal{UCI}} \sup_{P\in\mathcal{P}_\mathcal{A}} \frac{A_{P}(L_{\mathcal{C}})}{\max\{1,H(P)\}}   \\
          & \geq \sup_{P\in\mathcal{P}_\mathcal{A}} \inf_{\mathcal{C}\in\mathcal{UCI}} \frac{A_{P}(L_{\mathcal{C}})}{\max\{1,H(P)\}}   \\
          & \geq \sup_{P\in\mathcal{P}_\mathcal{A}(M)} \inf_{\mathcal{C}\in\mathcal{UCI}} \frac{A_{P}(L_{\mathcal{C}})}{\max\{1,H(P)\}}   \\
          & \geq \inf_{\mathcal{C}\in\mathcal{UCI}} \frac{A_{P_m}(L_{\mathcal{C}})}{\max\{1,H(P_m)\}}
\end{aligned}
\end{equation}
for all $2\leq m\in\mathbb{N}$.
When the probability distribution $P_m$ is given,\footnote{For any fixed $m$, the probability distribution $P_m$ is defined over a finite alphabet $\{1,2,\ldots,2^{m}+1 \}$.} the Huffman code~\cite{H52} is optimal; i.e., the average codeword length of the Huffman code is the shortest.
Therefore, we obtain
\begin{equation} \label{eq5}
\begin{aligned}
    \textcolor{red}{  C^{*} }&\textcolor{red}{  \geq \inf_{\mathcal{C}\in\mathcal{UCI}} \frac{A_{P_m}(L_{\mathcal{C}})}{\max\{1,H(P_m)\}} }  \\
           &  = \frac{1}{1+\textcolor{red}{h\big(\frac{1}{m}\big)}} \inf_{\mathcal{C}\in\mathcal{UCI}} A_{P_m}(L_{\mathcal{C}})   \\
          & \overset{(b)}>\frac{1}{1+\textcolor{red}{h\big(\frac{1}{m}\big)}}\Big[1\times\big(1-\frac{1}{m}\big)+2^m\times(m+1)\times\frac{1}{m\cdot2^{m}}  \Big]  \\
          & = \frac{2}{1+\textcolor{red}{h\big(\frac{1}{m}\big)}}
\end{aligned}
\end{equation}
for all $2\leq m\in\mathbb{N}$, where $(b)$ is true because $\mathcal{C}$ is a UCI, which cannot consider only the first $2^m+1$ codewords, and to make Kraft's inequality~\cite{kraft,EIT} hold, the average codeword length of the first $2^m+1$ codewords of a UCI code must therefore be strictly longer than that of a Huffman code.
From Equation~\eqref{eq5}, we have
\begin{equation} \label{eq6}
\begin{aligned}
    \textcolor{red}{ C^{*}}&\textcolor{red}{  \geq \sup_{2\leq m\in\mathbb{N}} \frac{2}{1+h\big(\frac{1}{m}\big)} }    \\
          & =\lim\limits_{m\rightarrow\infty} \frac{2}{1+\textcolor{red}{h\big(\frac{1}{m}\big)}}  \\
          & = 2.
\end{aligned}
\end{equation}
\end{proof}

{\color{blue}\section{Research Approaches for Determining the Upper Bound of $C^*$} \label{sec_idea}
In this section, we review the research approaches that have been used to address the upper bound of $C^*$ and provide inequalities that are tighter than those listed in previous work. On the basis of these inequalities, we propose an improved research approach.

\subsection{Review of research on the upper bound of $C^*$} \label{subsec}
The main work that has been conducted on the upper bound of $C^*$ is as follows.
In 1975, the expansion factor of the $\gamma$ code constructed by Elias~\cite{Elias75} was determined as $C_{\gamma}=3$, thus obtaining the first upper bound 3 for $C^*$.
Recently, the $\eta$~\cite{YL21} and $\iota$~\cite{YL22} codes constructed by Yan and Lin could achieve expansion factors $C_{\eta}=2.75$ and $C_{\iota}=2.5$, respectively; thus, the upper bound of $C^*$ was compressed to $2.5$.
The research approaches employed in these three studies were consistent, as detailed below.

First, for a given UCI $\mathcal{C}$, it must be proven the codeword length satisfies $L_\mathcal{C}(1)=1$ and $L_\mathcal{C}(a)\leq K_1+K_2\log_2a$ for all $2\leq a\in\mathcal{A}$, where $K_1$ and $K_2$ are two positive constants. Consequently, we obtain
\begin{equation*}
\begin{aligned}
           A_P(L_{\mathcal{C}}) & \leq P(1)+\sum_{a=2}^{\infty}P(a)\big(K_1+K_2\log_2a\big) \\
                                & =K_1+(1-K_1)P(1)+K_{2}\sum_{a=2}^{\infty}P(a)\log_2a.
\end{aligned}
\end{equation*}
Second, we need to use an inequality to obtain an upper bound for the expected value of $\log_2a$.
Elias~\cite{Elias75} used Wyner's inequality~\cite{Wyner72}:
\begin{equation}\label{eq34}
           \sum_{a=1}^{\infty}P(a)\log_2a=\sum_{a=2}^{\infty}P(a)\log_2a\leq H(P),
\end{equation}
and Yan \emph{et al.}~\cite{YL22} used
\begin{equation}\label{eq35}
           \sum_{a=2}^{\infty}P(a)\log_2a\leq H(P)+P(1)\log_2P(1).
\end{equation}
Next,
\begin{equation*}
          \frac{A_{P}(L_{\mathcal{C}})}{\max\{1,H(P)\}}\leq F_1\big(H(P),P(1)\big)
\end{equation*}
is obtained through the above inequality, where $F_1\big(H(P),P(1)\big)$ is a function of $H(P)$ and $P(1)$.
Finally, the upper bound of the function $F_1\big(H(P),P(1)\big)$ is analysed.

In short, three key points are involved in proving the upper bound of $C^*$.
\textcolor{red}{First, an inequality for the upper bound of the codeword length is needed, such as $L_\mathcal{C}(a)\leq K_1+K_2\log_2a$.
We call such an inequality a \emph{codeword length upper bound inequality}; such inequalities are related only to the codeword construction process and are independent of the underlying distribution.
Second, an inequality for the upper bound of the expected value of $\log_2a$ is needed, such as Inequalities~\eqref{eq34} and~\eqref{eq35}.
We call such inequalities \emph{probability inequalities for decreasing distributions}.
This type of inequality is the core of the proof.}
Third, the parameters of the upper bound function in the final analysis, for example, the parameters of $F_1$, are $H(P)$ and $P(1)$.
	
\subsection{The new probability inequality for decreasing distributions}
In this subsection, we provide a new probability inequality for decreasing distributions and
prove that this inequality is tighter than Inequalities~\eqref{eq34} and~\eqref{eq35};\footnote{Because  $H(P)+P(1)\log_2P(1)\leq H(P)$, Inequality \eqref{eq35} is tighter than Inequality \eqref{eq34}. Therefore, it is only necessary to show that the new inequality is tighter than Inequality \eqref{eq35}.} that is, it has a smaller upper bound.

Before the new inequality is given, we first define the relevant notations.
Let $T_P\triangleq \sup\{ a\in \mathcal{A} \mid P(a)>0 \}$ for any $P\in\mathcal{P}_\mathcal{A}$.
Note that if $P(a)>0$ for any $a\in \mathcal{A}$, then $T_{P}=+\infty$.
Let
\begin{equation} \label{eq40}
A_n\triangleq \log_{2}P(n)-\log_{2}(1-P(1))+\log_{2}n
\end{equation}
for all $n\in\mathbb{N}\cap[2,T_P]$.
When $T_{P}<+\infty$, $A_{n_0}$ is undefined for $n_0>T_{P}$.
However, owing to $0\log_{2}0\triangleq 0,\,$\footnote{In mathematics, $0\log_{2}0$ can be defined as $0\log_{2}0\triangleq \lim\limits_{x\rightarrow0^+}x\log_{2}x=0$.}
we can define
\[
  P(n_0)A_{n_0}=0\big[\log_{2}0-\log_{2}\big(1-P(1)\big)+\log_{2}n_0\big]\triangleq 0,
\]
for all $n_0>T_{P}$.
\begin{lemma}  \label{lemma4}
The inequality
\begin{equation} \label{eq36}
  \sum_{n=2}^{\infty}P(n)A_n \leq P(2)+ \sum_{n=2}^{\infty} P(n+1)\log_{2}\left(1+\frac{1}{n}\right)\left(1-\frac{1}{n}\right)^{n-1}
\end{equation}
holds for any $P\in\mathcal{P}_\mathcal{A}$.
Inequality~\eqref{eq36} can be equivalently rewritten as follows:
\begin{equation} \label{eq37}
  \sum_{n=2}^{\infty}P(n)\log_2n \leq H(P)-h\big(P(1)\big)+P(2)+ \sum_{n=2}^{\infty} P(n+1)\log_{2}\left(1+\frac{1}{n}\right)\left(1-\frac{1}{n}\right)^{n-1}.
\end{equation}
\end{lemma}
\begin{proof}
Let $B_n\triangleq \sum_{k=2}^{n}A_k$ for all $n\in\mathbb{N}\cap[2,T_P]$.
From the arithmetic mean--geometric mean inequality, we obtain
\[
  \left(\prod_{k=2}^{n}P(k)\right)^\frac{1}{n-1}\leq \frac{\sum_{k=2}^{n}P(k)}{n-1} \leq \frac{1-P(1)}{n-1}
\]
for any integer $n\geq2$. Furthermore, we have that
\begin{equation*}
\begin{aligned}
           B_n & = \log_{2}\prod_{k=2}^{n}P(k)-(n-1)\log_{2}\big(1-P(1)\big)+\log_{2}n!   \\
               & \leq  \log_{2}\left(\frac{1-P(1)}{n-1}\right)^{n-1}-\log_{2}\big(1-P(1)\big)^{n-1}+\log_{2}n!   \\
                          & = \log_{2}\frac{n!}{(n-1)^{n-1}}\,.
\end{aligned}
\end{equation*}
From Abel's formula, we obtain
\begin{equation} \label{eq25}
\begin{aligned}
         \sum_{n=2}^{m}P(n)A_n  & = \big(P(2)-P(3)\big)B_2+ \big(P(3)-P(4)\big)B_3+\cdots+ \big(P(m-1)-P(m)\big)B_{m-1}+P(m)B_m    \\
                                & = P(m)B_m + \sum_{n=2}^{m-1}\big(P(n)-P(n+1)\big)B_n           \\
                                & \leq  P(m)\log_{2}\frac{m!}{(m-1)^{m-1}} + \sum_{n=2}^{m-1}\big(P(n)-P(n+1)\big) \log_{2}\frac{n!}{(n-1)^{n-1}}           \\
                                & = P(2)+\sum_{n=2}^{m-1}P(n+1)\left(\log_{2}\frac{(n+1)!}{n^{n}}- \log_{2}\frac{n!}{(n-1)^{n-1}}  \right)\\
                                & = P(2)+\sum_{n=2}^{m-1}P(n+1)\log_{2}\left(\frac{(n+1)!}{n^{n}}\cdot\frac{(n-1)^{n-1}}{n!}\right)  \\
                                & = P(2)+\sum_{n=2}^{m-1}P(n+1)\log_{2}\left(1+\frac{1}{n}\right)\left(1-\frac{1}{n}\right)^{n-1}
\end{aligned}
\end{equation}
for all $m\in\mathbb{N}\cap[2,T_P]$.
We consider the following two cases.
\begin{enumerate}
\item When $T_P =+\infty$, owing to~Equation \eqref{eq25},
\begin{equation*}
\begin{aligned}
         \sum_{n=2}^{\infty}P(n)A_n  & = \lim_{m\rightarrow\infty}\sum_{n=2}^{m}P(n)A_n    \\
         & \leq \lim_{m\rightarrow\infty}\left[ P(2)+\sum_{n=2}^{m-1}P(n+1)\log_{2}\left(1+\frac{1}{n}\right)\left(1-\frac{1}{n}\right)^{n-1}   \right]  \\
         & = P(2)+\lim_{m\rightarrow\infty}\left[\sum_{n=2}^{m-1}P(n+1)\log_{2}\left(1+\frac{1}{n}\right)\left(1-\frac{1}{n}\right)^{n-1}   \right]  \\
         & = P(2)+ \sum_{n=2}^{\infty} P(n+1)\log_{2}\left(1+\frac{1}{n}\right)\left(1-\frac{1}{n}\right)^{n-1}.
\end{aligned}
\end{equation*}
\item When $T_P <+\infty$, we consider~Equation \eqref{eq25} by taking $m=T_P$; then,
\[
  \sum_{n=2}^{T_P}P(n)A_n \leq P(2)+\sum_{n=2}^{T_P-1}P(n+1)\log_{2}\left(1+\frac{1}{n}\right)\left(1-\frac{1}{n}\right)^{n-1}.
\]
Furthermore, we obtain
\begin{equation*}
\begin{aligned}
         \sum_{n=2}^{\infty}P(n)A_n  & = 0+\sum_{n=2}^{T_P}P(n)A_n  \\
         & \leq  P(2)+ \sum_{n=2}^{T_P-1}P(n+1)\log_{2}\left(1+\frac{1}{n}\right)\left(1-\frac{1}{n}\right)^{n-1}  \\
         & = P(2)+ \sum_{n=2}^{\infty} P(n+1)\log_{2}\left(1+\frac{1}{n}\right)\left(1-\frac{1}{n}\right)^{n-1}.
\end{aligned}
\end{equation*}
\end{enumerate}  
\end{proof}
At the end of this subsection, we prove that Inequality~\eqref{eq37} is tighter than Inequality~\eqref{eq35}.
However, we first need to prove an auxiliary lemma.
For convenience, let $C_n\triangleq\left(1+\frac{1}{n}\right)\left(1-\frac{1}{n}\right)^{n-1}$ for all $2\leq n\in\mathbb{N}$ in this \textcolor{red}{paper}.
\begin{lemma} \label{lemma5}
The sequence $\{C_n\}_{n=2}^\infty$ is strictly monotonically decreasing.
\end{lemma}
\begin{proof}
First, $1+\frac{1}{n}$ is clearly strictly monotonically decreasing with respect to $n$.
Second, the sequence
$$\left\{\left(1-\frac{1}{n}\right)^{n-1}=\left(\frac{n-1}{n}\right)^{n-1}\right\}_{n=2}^\infty$$
being strictly monotonically decreasing is equivalent to the sequence  $\left\{\left(\frac{n}{n-1}\right)^{n-1}\right\}_{n=2}^\infty$ being strictly monotonically increasing.
From the arithmetic mean--geometric mean inequality, we have that
\begin{equation*}
\begin{aligned}
         \left(\frac{n}{n-1}\right)^{n-1} & = \underbrace{\frac{n}{n-1}\cdots\frac{n}{n-1}}_{n-1} \times 1 \\
                                          & < \left[ \frac{\frac{n}{n-1}\times(n-1)+1}{n}\right]^n   \\
         & = \left(\frac{n+1}{n}\right)^{n}.
\end{aligned}
\end{equation*}
In conclusion, $\{C_n\}_{n=2}^\infty$ is strictly monotonically decreasing.
\end{proof}
\begin{lemma}
The inequality
\begin{equation} \label{eq38}
  H(P)-h\big(P(1)\big)+P(2)+ \sum_{n=2}^{\infty} P(n+1)\log_{2}\left(1+\frac{1}{n}\right)\left(1-\frac{1}{n}\right)^{n-1}\leq H(P)+P(1)\log_2P(1)
\end{equation}
holds for any $P\in\mathcal{P}_\mathcal{A}$. Inequality~\eqref{eq38} can be equivalently rewritten as follows:
\begin{equation} \label{eq39}
   \big(1-P(1)\big)\log_2\big(1-P(1)\big)+P(2)+ \sum_{n=2}^{\infty} P(n+1)\log_{2}\left(1+\frac{1}{n}\right)\left(1-\frac{1}{n}\right)^{n-1}\leq 0.
\end{equation}
\end{lemma}
\begin{proof}
We prove that Inequality~\eqref{eq39} holds. Based on Lemma~\ref{lemma5}, we obtain
\begin{equation*}
\begin{aligned}
       &  \quad \big(1-P(1)\big)\log_2\big(1-P(1)\big)+P(2)+ \sum_{n=2}^{\infty} P(n+1)\log_{2}\left(1+\frac{1}{n}\right)\left(1-\frac{1}{n}\right)^{n-1} \\
          & = P(2)\log_{2}\left[2\big(1-P(1)\big)\right]+\sum_{n=2}^{\infty} P(n+1)\log_{2}\left[\big(1-P(1)\big)\left(1+\frac{1}{n}\right)\left(1-\frac{1}{n}\right)^{n-1}\right] \\
         & \leq P(2)\log_{2}\left[2\big(1-P(1)\big)\right]+\sum_{n=2}^{\infty} P(n+1)\log_{2}\left[\big(1-P(1)\big)\left(1+\frac{1}{2}\right)\left(1-\frac{1}{2}\right)\right]  \\
         & = P(2)\log_{2}\left[2\big(1-P(1)\big)\right]+\big(1-P(1)-P(2)\big)\log_{2}\left[\frac{3}{4}\big(1-P(1)\big)\right]  \\
         & = \big(1-P(1)\big)\log_{2}\left[\frac{3}{4}\big(1-P(1)\big)\right] +P(2)\log_2\frac{8}{3}.
\end{aligned}
\end{equation*}
Consider the following two cases.
\begin{enumerate}
\item When $0<P(1)<\frac{1}{2}$, we further obtain
\begin{equation*}
\begin{aligned}
     \big(1-P(1)\big)\log_{2}\left[\frac{3}{4}\big(1-P(1)\big)\right] +P(2)\log_2\frac{8}{3} &\leq  \big(1-P(1)\big)\log_{2}\left[\frac{3}{4}\big(1-P(1)\big)\right] +P(1)\log_2\frac{8}{3}   \\
     & = \big(1-P(1)\big)\log_{2}\Big(1-P(1)\Big)+\log_2\frac{3}{4}+P(1)\log_2\frac{32}{9}  \\
     &  \overset{(a)}< \big(1-\frac{1}{2}\big)\log_{2}\Big(1-\frac{1}{2}\Big)+\log_2\frac{3}{4}+\frac{1}{2}\log_2\frac{32}{9} \\
     & =0,
\end{aligned}
\end{equation*}
where $(a)$ is true because $\big(1-x\big)\log_{2}\Big(1-x\Big)+\log_2\frac{3}{4}+x\log_2\frac{32}{9}$ is strictly monotonically increasing over the interval $(0, \frac{1}{2})$.
\item When $\frac{1}{2}\leq P(1)<1$, we further obtain
\begin{equation*}
\begin{aligned}
     \big(1-P(1)\big)\log_{2}\left[\frac{3}{4}\big(1-P(1)\big)\right] +P(2)\log_2\frac{8}{3} &\leq  P(2)\log_{2}\left[\frac{3}{4}\big(1-P(1)\big)\right] +P(2)\log_2\frac{8}{3}   \\
     & = P(2)\log_{2}\left[2\big(1-P(1)\big)\right]   \\
     & \leq P(2)\log_{2}\left[2\Big(1-\frac{1}{2}\Big)\right]   \\
     & =0.
\end{aligned}
\end{equation*}
\end{enumerate}
\end{proof}

\subsection{Our research approach concerning the upper bound of $C^*$}
In this subsection, we present our approach for studying the upper bound of $C^*$ based on
the tighter probability inequality proposed in Lemma~\ref{lemma4}.

Inequality~\eqref{eq36} can equivalently be rewritten as follows:
\begin{small}
\begin{equation} \label{eq41}
  \sum_{a=3}^{\infty}\!P(a)\big[\log_{2}a-\log_{2}(1-P(1))\big]\! \leq\! H(P)+P(1)\log_{2}P(1)+P(2)\log_{2}\big(1\!-\!P(1)\big)+ \sum_{a=2}^{\infty} \!P(a+1)\log_{2}\left(1\!+\!\frac{1}{a}\right)\left(1\!-\!\frac{1}{a}\right)^{a\!-\!1}.
\end{equation}
\end{small}
On the basis of the research approaches summarized in Subsection~\ref{subsec},
we improve upon the three key points involved in the proof of the upper bound of $C^*$ according to Inequality~\eqref{eq41}.
First, the codeword length upper bound inequality is improved from the traditional $L_\mathcal{C}(a)\leq K_1+K_2\log_2a$ to
\begin{equation} \label{eq42}
  L_\mathcal{C}(a)\leq F_2\big(P(1)\big)\big[\log_{2}a-\log_{2}(1-P(1))\big],
\end{equation}
where $F_2\big(P(1)\big)$ represents a piecewise function.
This idea is inspired by the form on the left-hand side of Inequality~\eqref{eq41}.
At this time, Inequality~\eqref{eq42} can be more tightly scaled according to different values of $P(1)$.
Second, we use Inequality~\eqref{eq41}, which is tighter than Inequality~\eqref{eq35}, as the probability inequality for decreasing distributions. Since an infinite summation term is located on the right-hand side of Inequality~\eqref{eq41}, which is inconvenient for the analysis, an appropriate scaling process will also be performed at the end of this subsection.
Finally, owing to the form on the right-hand side of Inequality~\eqref{eq41}, we should ultimately obtain
\begin{equation*}
          \frac{A_{P}(L_{\mathcal{C}})}{\max\{1,H(P)\}}\leq F_3\big(H(P),P(1),P(2)\big)
\end{equation*}
where $F_3\big(H(P),P(1),P(2)\big)$ is a function of $H(P)$, $P(1)$ and $P(2)$.\footnote{In the subsequent proof described in this paper, using Equation~\eqref{eq22}, $F_3\big(H(P),P(1),P(2)\big)$ can be further simplified to $F_3\big(P(1),P(2)\big)$; that is, $F_3$ is independent of $H(P)$.}
This upper bound function $F_3$ is related not only to $P(1)$ but also to $P(2)$, making the analysis more precise.

The codeword length upper bound inequality is given after the code construction procedure.
At the end of this subsection, we provide a probability inequality for decreasing distributions that is convenient for analysis purposes; its form is roughly
\begin{equation*}
          \sum_{a=3}^{\infty}P(a)\left[\log_{2}a-\log_{2}\big(1-P(1)\big)\right]\leq H(P)+F_4(P(1))+P(1)F_5(P(1))+P(2)F_6(P(1)).
\end{equation*}
The specific form is given by the following lemma.
\begin{lemma}  \label{lemma6}
The probability inequality
\begin{equation*}
\begin{aligned}
         \sum_{a=3}^{\infty}P(a)\left[\log_{2}a-\log_{2}\big(1-P(1)\big)\right]
            & \leq H(P)+\log_{2}C_{L}+P(1)\big(\log_{2}P(1)-\log_{2}C_{L}\big)   \\
            &\quad + P(2)\left[ \log_{2}\big(1-P(1)\big)+\sum_{a=2}^{L-1}\log_{2}C_{a} -(L-1)\log_{2}C_{L}  \right]
\end{aligned}
\end{equation*}
holds for any $P\in\mathcal{P}_\mathcal{A}$ and all $2\leq L\in\mathbb{N}$. In particular, when $L=2$, 
\begin{equation*}
\begin{aligned}
         \sum_{a=3}^{\infty}P(a)\left[\log_{2}a-\log_{2}\big(1-P(1)\big)\right]
            & \leq H(P)+\log_{2}\frac{3}{4}+P(1)\left(\log_{2}P(1)-\log_{2}\frac{3}{4}\right)   \\
            &\quad + P(2)\left[ \log_{2}\big(1-P(1)\big)-\log_{2}\frac{3}{4} \right]
\end{aligned}
\end{equation*}
for any $P\in\mathcal{P}_\mathcal{A}$.
\end{lemma}
\begin{proof}
Based on Lemma~\ref{lemma4} and Lemma~\ref{lemma5}, we obtain
\begin{equation*}
\begin{aligned}
         \sum_{n=2}^{\infty}P(n)A_n  & \leq P(2)+ \sum_{a=2}^{\infty} P(a+1)\log_{2}C_{a} \\
         & \leq  P(2)+ \sum_{a=2}^{L-1}P(a+1)\log_{2}C_{a}+\sum_{a=L}^{\infty}P(a+1)\log_{2}C_{L} \\
         & = P(2)+ \sum_{a=2}^{L-1}P(a+1)\log_{2}C_{a}+\left(1-\sum_{a=1}^{L}P(a)\right)\log_{2}C_{L} \\
         & =\big(1-P(1)\big)\log_{2}C_{L}+P(2)(1-\log_{2}C_{L})+\sum_{a=2}^{L-1}P(a+1)(\log_{2}C_{a}-\log_{2}C_{L})  \\
         & \leq \big(1-P(1)\big)\log_{2}C_{L}+P(2)(1-\log_{2}C_{L})+\sum_{a=2}^{L-1}P(2)(\log_{2}C_{a}-\log_{2}C_{L}) \\
         & = \big(1-P(1)\big)\log_{2}C_{L}+P(2)\left(1+\sum_{a=2}^{L-1}\log_{2}C_{a}-(L-1)\log_{2}C_{L}\right).
\end{aligned}
\end{equation*}
Furthermore, we have that
\begin{equation*}
\begin{aligned}
         \sum_{a=2}^{\infty}P(a)\left[\log_{2}a-\log_{2}\big(1-P(1)\big)\right] & \leq H(P)+P(1)\log_{2}P(1)+ \big(1-P(1)\big)\log_{2}C_{L}  \\
            & \quad  +P(2)\left(1+\sum_{a=2}^{L-1}\log_{2}C_{a}-(L-1)\log_{2}C_{L}\right),\\
\end{aligned}
\end{equation*}
and hence, we can obtain
\begin{equation*}
\begin{aligned}
         \sum_{a=3}^{\infty}P(a)\left[\log_{2}a-\log_{2}\big(1-P(1)\big)\right] & \leq H(P)+P(1)\log_{2}P(1)+ \big(1-P(1)\big)\log_{2}C_{L}  \\
            & \quad  +P(2)\left(1+\sum_{a=2}^{L-1}\log_{2}C_{a}-(L-1)\log_{2}C_{L}-1+\log_{2}\big(1-P(1)\big)\right),\\
           & \leq H(P)+\log_{2}C_{L}+P(1)\big(\log_{2}P(1)-\log_{2}C_{L}\big)   \\
            &\quad + P(2)\left[ \log_{2}\big(1-P(1)\big)+\sum_{a=2}^{L-1}\log_{2}C_{a} -(L-1)\log_{2}C_{L}  \right].
\end{aligned}
\end{equation*}
\end{proof}}
\section{A Tighter Upper Bound for $C^{*}$ and the Corresponding Codeword Length Construction Schemes}\label{newcode}
In this section, we construct a new class of UCIs, which is termed the $\nu$ code.
The $\nu$ code has an expansion factor $C_{\nu}=2.0386$.
This leads to
\begin{equation}  \label{eq7}
\begin{aligned}
     C^{*}&= \inf_{\mathcal{C}\in\mathcal{UCI}}  C_{\mathcal{C}}^{*} \\
          &\leq C_{\nu}^{*}   \\
          &= \inf \left\{C_{\nu} \Big| \frac{A_{P}(L_{\mathcal{C}})}{\max\{1,H(P)\}}\leq C_{\nu} \mbox{ for }\forall P\in\mathcal{P}_\mathcal{A} \right\}       \\
          & \leq 2.0386.
\end{aligned}
\end{equation}
We first present the length function of the $\nu$ code and then provide its specific construction process.

\subsection{The length function of the $\nu$ code and its existence}
The goal of constructing the $\nu$ code is to ensure that the length of its codewords can satisfy
the codeword length upper bound inequality in the form of Equation~\eqref{eq42}.
Therefore, we directly define the codeword length of the $\nu$ code as follows.
For all $a\in\mathcal{A}$, the length function $L_{\nu}(a)$ is defined as
\begin{equation} \label{eq43}
  L_{\nu}(a) \triangleq L_{\delta}(a)+\Delta(a),
\end{equation}
where
\begin{equation*}
\Delta(a)\triangleq \left\{\begin{array}{lllll}
2,                  &\text{if }   a=6,7\text{ ,}\\
1,         &\text{if }   a=3,4,5\text{ ,}   \\
-1,           &\text{if } a=2\text{ or } \lfloor\log_{2}a\rfloor\in S_{\nu1}\text{ ,}\\
-2,           &\text{if } \lfloor\log_{2}a\rfloor\in S_{\nu2}\text{ ,}\\
0,                          &\text{otherwise,}   \\
\end{array}\right.
\end{equation*}
\begin{equation*}
\begin{aligned}
  S_{\nu1} & \triangleq \{7,15,16,17,\ldots,24,37,38,39,\ldots,50,68,69,70,\ldots,84\},  \\
  S_{\nu2} & \triangleq \{31,32,33,34,35,36,63,64,65,66,67\}.
\end{aligned}
\end{equation*}
Next, we prove the existence of a prefix code with such codeword lengths.
More precisely, we prove the following theorem.
\begin{theorem}
There exist prefix codes whose codeword lengths are $L_{\nu}(1), L_{\nu}(2),\ldots,L_{\nu}(a),\ldots.$
\end{theorem}
\begin{proof}
From Kraft's inequality~\cite{kraft,EIT}, this is equivalent to proving that $\sum_{a=1}^{\infty}2^{-L_{\nu}(a)}\leq 1$.
Our objective is to prove that
\begin{equation}  \label{eq99}
\sum_{a=1}^{\infty}2^{-L_{\nu}(a)}= 1
\end{equation}
holds. First, we calculate
\begin{equation}   \label{eq98}
\begin{aligned}
\sum_{a=1}^{\infty}2^{-L_{\delta}(a)} & =\sum_{a=1}^{\infty}2^{-(1+\lfloor\log_{2}a\rfloor+2\lfloor\log_{2}(1+\lfloor\log_{2}a\rfloor)\rfloor)}  \\
                                      & \overset{(b)}=\sum_{t=0}^{\infty}2^{t}\cdot2^{-(1+t+2\lfloor\log_{2}(1+t)\rfloor)} \\
                                      & =2^{-1} \sum_{t=0}^{\infty} \left(\frac{1}{4}\right)^{\lfloor\log_{2}(1+t)\rfloor}  \\
                                      & \overset{(c)}=2^{-1} \sum_{s=0}^{\infty} 2^{s} \left(\frac{1}{4}\right)^{s}  \\
                                      & =2^{-1} \sum_{s=0}^{\infty} 2^{-s}   \\
                                      & =1,
\end{aligned}
\end{equation}
where $(b)$ and $(c)$ are the results of the substitutions $t=\lfloor\log_{2}a\rfloor$ and $s=\lfloor\log_{2}(1+t)\rfloor$, respectively.
Due to~Equation \eqref{eq98} and the definition of $L_{\nu}(a)$, proving~Equation \eqref{eq99} is equivalent to proving that
\begin{equation}  \label{eq97}
\sum_{a=2}^{7}2^{-L_{\delta}(a)}+\sum_{\lfloor\log_{2}a\rfloor\in S_{\nu1}\cup S_{\nu2}}2^{-L_{\delta}(a)} = \sum_{a=2}^{7}2^{-L_{\nu}(a)}+\sum_{\lfloor\log_{2}a\rfloor\in S_{\nu1}\cup S_{\nu2}}2^{-L_{\nu}(a)}.
\end{equation}
Afterward, we obtain
\begin{equation*}
\begin{aligned}
& \quad    \sum_{a=2}^{7}2^{-L_{\delta}(a)}+\sum_{\lfloor\log_{2}a\rfloor\in S_{\nu1}\cup S_{\nu2}}2^{-L_{\delta}(a)}  \\
                             & =2\cdot2^{-4}+4\cdot2^{-5}+2^{7}\cdot2^{-14}+\sum_{t=15}^{24}2^{t}\cdot2^{-(1+t+2\lfloor\log_{2}(1+t)\rfloor)}   \\
                             & \quad +\sum_{t=31}^{50}2^{t}\cdot2^{-(1+t+2\lfloor\log_{2}(1+t)\rfloor)}+\sum_{t=63}^{84}2^{t}\cdot2^{-(1+t+2\lfloor\log_{2}(1+t)\rfloor)} \\
                             & =2^{-2}+2^{-7}+\sum_{t=15}^{24}2^{-9}+\sum_{t=31}^{50}2^{-11}+\sum_{t=63}^{84}2^{-13} \\
                             & =\frac{1187}{4096},
\end{aligned}
\end{equation*}
and
\begin{equation*}
\begin{aligned}
& \quad    \sum_{a=2}^{7}2^{-L_{\nu}(a)}+\sum_{\lfloor\log_{2}a\rfloor\in S_{\nu1}\cup S_{\nu2}}2^{-L_{\nu}(a)} \\
                             & =2^{-3}+2^{-5}+2^{-5}+2^{-6}+2^{7}\cdot2^{-13}+\sum_{t=15}^{24}2^{t}\cdot2^{-(t+2\lfloor\log_{2}(1+t)\rfloor)}
                               +\sum_{t=31}^{36}2^{t}\cdot2^{-(-1+t+2\lfloor\log_{2}(1+t)\rfloor)}   \\
                             & \quad +\sum_{t=37}^{50}2^{t}\cdot2^{-(t+2\lfloor\log_{2}(1+t)\rfloor)}+\sum_{t=63}^{67}2^{t}\cdot2^{-(-1+t+2\lfloor\log_{2}(1+t)\rfloor)}
                             +\sum_{t=68}^{84}2^{t}\cdot2^{-(t+2\lfloor\log_{2}(1+t)\rfloor)} \\
                             & =2^{-3}+3\cdot2^{-5}+\sum_{t=15}^{24}2^{-8}+\sum_{t=31}^{36}2^{-9}+\sum_{t=37}^{50}2^{-10}+\sum_{t=63}^{67}2^{-11}+\sum_{t=68}^{84}2^{-12} \\
                             & =\frac{1187}{4096}.
\end{aligned}
\end{equation*}
Thus,~Equation \eqref{eq97} holds. The proof is complete.
\end{proof}
Finally, the main theorem of this paper is presented here, with its proof provided in the next section.
\begin{theorem}\label{thm5}
The $\nu$ code is a UCI and has an expansion factor $C_{\nu}=2.0386$; that is,
\begin{equation*}
\frac{A_{P}(L_{\nu})}{\max\{1,H(P)\}}\leq 2.0386
\end{equation*}
for all $P\in\mathcal{P}_\mathcal{A}$.
Thus, $C^{*}\leq 2.0386$.
\end{theorem}
{\color{red}\subsection{The specific construction of the $\nu$ code}
In this subsection, we construct the $\nu$ code using two methods.

Because the length function of the $\nu$ code is similar to that of the Elias $\delta$ code, the first construction scheme is similar to the way in which Elias constructed the $\delta$ code. Let us first construct an auxiliary $\Delta\gamma$ code.
Since the newly constructed code is obtained from a change in the Elias $\gamma$ code, it is called the $\Delta\gamma$ code.
Notably, the $\Delta\gamma$ code is encoded starting from $0$.
The encoding strategy of the $\Delta\gamma$ code is as follows.
The first 86 codewords of the $\Delta\gamma$ code are shown in Table~\ref{tab10},
and from the 87-th codeword of the encoding onward, $\Delta\gamma(a)=\gamma(a)$ for all $a\geq86$.
The last column of Table~\ref{tab10} shows the relationship between the $\gamma$ code and the $\Delta\gamma$ code, where $\gamma(32)\cap\gamma(33)$ denotes the shortening of $\gamma(32)=00000100000$ and $\gamma(33)=00000100001$ refers to the common longest prefix $0000010000$.

\begin{table*}[!htbp]
\centering
\color{red}
\caption{The first 85 codewords of the $\gamma$ code and the first 86 codewords of the $\Delta\gamma$ code}\label{tab10}
\begin{tabular}{c|c|c|c||c|c|c|c}
\hline \hline
$n$  &  $\gamma$ code  & $\Delta\gamma$ code & relationship & $n$  &  $\gamma$ code  & $\Delta\gamma$ code & relationship \\
\hline \hline
$0$   &     ----         &    1         &   $\gamma(1)$       &  $43$     & 000001 01011  & 000011110 1   &   $\gamma(30)1$  \\
$1$   &      1           &   010        &   $\gamma(2)$       &  $44$     & 000001 01100  & 000011111 0   &  $\gamma(31)0$  \\
$2$   &     01 0         &  011 0       &   $\gamma(3)$0      &  $45$     & 000001 01101  & 000011111 1   &  $\gamma(31)1$  \\
$3$   &     01 1         &  011 100     &   $\gamma(3)$100    &  $46$     & 000001 01110  & 0000010000    &  $\gamma(32)\cap\gamma(33)$  \\
$4$   &     001 00       &  011 11      &   $\gamma(3)$11     &  $47$     & 000001 01111  & 0000010001    &  $\gamma(34)\cap\gamma(35)$  \\
$5$   &     001 01       &  00100       &   $\gamma(4)$       &  $48$     & 000001 10000  & 0000010010    &  $\gamma(36)\cap\gamma(37)$  \\
$6$   &     001 10       &  00101       &   $\gamma(5)$       &  $49$     & 000001 10001  & 0000010011    &  $\gamma(38)\cap\gamma(39)$  \\
$7$   &     001 11       &  00110       &   $\gamma(6)$       &  $50$     & 000001 10010  & 0000010100    &  $\gamma(40)\cap\gamma(41)$  \\
$8$   &     0001 000     &  011 101     &   $\gamma(3)$101    &  $51$     & 000001 10011  & 0000010101    &  $\gamma(42)\cap\gamma(43)$  \\
$9$   &     0001 001     &  00111 00    &   $\gamma(7)$00     &  $52$     & 000001 10100  & 00000101100   &  $\gamma(44)$  \\
$10$  &     0001 010     &  00111 01    &   $\gamma(7)$01     &  $53$     & 000001 10101  & 00000101101   &  $\gamma(45)$  \\
$11$   &    0001 011     &  00111 10    &   $\gamma(7)$10     &  $54$     & 000001 10110  & 00000101110   &  $\gamma(46)$  \\
$12$   &    0001 100     &  00111 11    &   $\gamma(7)$11     &  $55$     & 000001 10111  & 00000101111   &  $\gamma(47)$  \\
$13$   &    0001 101     &  0001000     &   $\gamma(8)$       &  $56$     & 000001 11000  & 00000110000   &  $\gamma(48)$  \\
$14$   &    0001 110     &  0001001     &   $\gamma(9)$       &  $57$     & 000001 11001  & 00000110001   &  $\gamma(49)$  \\
$15$   &    0001 111     &  0001010     &   $\gamma(10)$      &  $58$     & 000001 11010  & 00000110010   &  $\gamma(50)$  \\
$16$   &    00001 0000   &  0001011 0   &   $\gamma(11)0$     &  $59$     & 000001 11011  & 00000110011   &  $\gamma(51)$  \\
$17$   &    00001 0001   &  0001011 1   &   $\gamma(11)1$     &  $60$     & 000001 11100  & 00000110100   &  $\gamma(52)$  \\
$18$   &    00001 0010   &  0001100 0   &   $\gamma(12)0$     &  $61$     & 000001 11101  & 00000110101   &  $\gamma(53)$  \\
$19$   &    00001 0011   &  0001100 1   &   $\gamma(12)1$     &  $62$     & 000001 11110  & 00000110110   &  $\gamma(54)$  \\
$20$   &    00001 0100   &  0001101 0   &   $\gamma(13)0$     &  $63$     & 000001 11111  & 00000110111   &  $\gamma(55)$  \\
$21$   &    00001 0101   &  0001101 1   &   $\gamma(13)1$     &  $64$     & 0000001 000000& 00000111000   &  $\gamma(56)$  \\
$22$   &    00001 0110   &  0001110 0   &   $\gamma(14)0$     &  $65$     & 0000001 000001& 00000111001   &  $\gamma(57)$  \\
$23$   &    00001 0111   &  0001110 1   &   $\gamma(14)1$     &  $66$     & 0000001 000010& 00000111010   &  $\gamma(58)$  \\
$24$   &    00001 1000   &  0001111 0   &   $\gamma(15)0$     &  $67$     & 0000001 000011& 00000111011   &  $\gamma(59)$  \\
$25$   &    00001 1001   &  0001111 1   &   $\gamma(15)1$     &  $68$     & 0000001 000100& 00000111100   &  $\gamma(60)$  \\
$26$   &    00001 1010   &  000010000   &   $\gamma(16)$      &  $69$     & 0000001 000101& 00000111101 0 &  $\gamma(61)0$  \\
$27$   &    00001 1011   &  000010001   &   $\gamma(17)$      &  $70$     & 0000001 000110& 00000111101 1 &  $\gamma(61)1$  \\
$28$   &    00001 1100   &  000010010   &   $\gamma(18)$      &  $71$     & 0000001 000111& 00000111110 0 &  $\gamma(62)0$  \\
$29$   &    00001 1101   &  000010011   &   $\gamma(19)$      &  $72$     & 0000001 001000& 00000111110 1 &  $\gamma(62)1$  \\
$30$   &    00001 1110   &  000010100   &   $\gamma(20)$      &  $73$     & 0000001 001001& 00000111111 0 &  $\gamma(63)0$  \\
$31$   &    00001 1111   &  000010101   &   $\gamma(21)$      &  $74$     & 0000001 001010& 00000111111 1 &  $\gamma(63)1$  \\
$32$   &    000001 00000 &  000010110   &   $\gamma(22)$      &  $75$     & 0000001 001011& 000000100000  &  $\gamma(64)\cap\gamma(65)$  \\
$33$   &    000001 00001 &  000010111   &   $\gamma(23)$      &  $76$     & 0000001 001100& 000000100001  & $\gamma(66)\cap\gamma(67)$  \\
$34$   &    000001 00010 &  000011000   &   $\gamma(24)$      &  $77$     & 0000001 001101& 000000100010  &  $\gamma(68)\cap\gamma(69)$  \\
$35$   &    000001 00011 &  000011001   &   $\gamma(25)$      &  $78$     & 0000001 001110& 000000100011  &  $\gamma(70)\cap\gamma(71)$  \\
$36$   &    000001 00100 &  000011010   &   $\gamma(26)$      &  $79$     & 0000001 001111& 000000100100  &  $\gamma(72)\cap\gamma(73)$  \\
$37$   &    000001 00101 &  000011011   &   $\gamma(27)$      &  $80$     & 0000001 010000& 000000100101  &  $\gamma(74)\cap\gamma(75)$  \\
$38$   &    000001 00110 &  000011100 0 &   $\gamma(28)0$     &  $81$     & 0000001 010001& 000000100110  &  $\gamma(76)\cap\gamma(77)$  \\
$39$   &    000001 00111 &  000011100 1 &   $\gamma(28)1$     &  $82$     & 0000001 010010& 000000100111  &  $\gamma(78)\cap\gamma(79)$  \\
$40$   &    000001 01000 &  000011101 0 &   $\gamma(29)0$     &  $83$     & 0000001 010011& 000000101000  &  $\gamma(80)\cap\gamma(81)$  \\
$41$   &    000001 01001 &  000011101 1 &   $\gamma(29)1$     &  $84$     & 0000001 010100& 000000101001  &  $\gamma(82)\cap\gamma(83)$  \\
$42$   &    000001 01010 &  000011110 0 &   $\gamma(30)0$     &  $85$     & 0000001 010101& 000000101010  &  $\gamma(84)\cap\gamma(85)$  \\
\hline \hline
\end{tabular}
\end{table*}

$\nu$ code~\cite{Elias75}: $\mathcal{A}\rightarrow \{0,1\}^{*}$ can be represented as
\begin{equation} \label{eq47}
\nu(a)\triangleq \left\{\begin{array}{llllllll}
1,                  &\text{if }   a=1\text{ ,}\\
010,         &\text{if }   a=2\text{ ,}   \\
01100,           &\text{if } a=3\text{ ,}\\
011010,                  &\text{if }   a=4\text{ ,}\\
011011,         &\text{if }   a=5\text{ ,}   \\
0111000,           &\text{if } a=6\text{ ,}\\
0111001,          &\text{if } a=7\text{ ,}\\
\Delta\gamma(|\beta(a)|)[\beta(a)],     &\text{if } a\geq8\text{ .}\\
\end{array}\right.
\end{equation}
From the construction of the $\nu$ code, it can be seen that the auxiliary $\Delta\gamma$ code is only activated when $a\geq8$,
which indicates that the $\Delta\gamma$ code starts affecting the construction process of the $\nu$ code from the $5=1+(\log_{2}8+1)$-th codeword $01111$.
It can be verified that the constructed $\nu$ code satisfies Equation~\eqref{eq43}.
From the construction schemes of the $\Delta\gamma$ code and $\nu$ code, Equation~\eqref{eq43} is clearly satisfied when $a\leq7$ or $a\geq2^{85}$.
Therefore, it is only necessary to verify these $82$ values when $a=2^t$, where $t\in\{3,4,\ldots,84\}$.
For example, when $a=2^{33}$, we obtain
\[
 L_\nu(2^{33})=1+33+2\lfloor\log_{2}(1+33)\rfloor-2=42
\]
due to Equation~\eqref{eq43}, and $\nu(2^{33})=000011000\underbrace{00\cdots 0}_{33}$ based on the construction of the $\nu$ code.
This verification is not obvious. Therefore, we provide a second construction process.

The second construction process originates from alphabetic codes~\cite{91NN,91Yeung,2024code}.
Nakatsu~\cite{91NN} proposed a scheme for constructing alphabetic codes, which is described as follows.\footnote{Although the scheme developed in~\cite{91NN} targets a finite alphabet, any two codewords constructed by the scheme are not prefixes of each other; thus, extending this alphabet to a countably infinite alphabet is feasible.}
Let $l_n$ denote the length of the $n$-th codeword.
Let $\alpha_n\triangleq\min\{l_{n-1},l_n\}$ for all $2\leq n\in\mathbb{N}$.
We define $trunc(x,n)\triangleq\frac{\lfloor2^{n}x\rfloor}{2^n}$, which denotes the first $n$ bits after the decimal point of the binary representation of the real number $x<1$.
The coding scheme takes the first $l_n$ bits after the decimal point of the binary representation of $sum(n)$ as the $n$-th codeword for all $n\in\mathbb{N}$, where
\begin{equation*}
sum(n)\triangleq \left\{\begin{array}{ll}
0,                  &\text{if }   n=1\text{ ,}\\
trunc\big(sum(n-1),\alpha_n\big)+\Big(\dfrac{1}{2}\Big)^{\alpha_n},       &\text{if } n\geq2\text{ .}\\
\end{array}\right.
\end{equation*}

In this paper, the codeword length satisfies $l_{n-1}\leq l_n$ for all $2\leq n\in\mathbb{N}$; therefore, we can obtain that $\alpha_n=l_{n-1}$ for all $2\leq n\in\mathbb{N}$ and
\begin{equation*}
sum(n)= \left\{\begin{array}{ll}
0,                  &\text{if }   n=1\text{ ,}\\
\sum_{i=1}^{n-1}\Big(\dfrac{1}{2}\Big)^{l_i},       &\text{if } n\geq2\text{ .}\\
\end{array}\right.
\end{equation*}
Taking the $\nu$ code as an example, Equation~\eqref{eq43} needs to be satisfied.
Therefore, we have that $l_1=1$, $l_2=3$, $l_3=5$, $l_4=6$, $l_5=6$, $l_6=7$, $l_7=8$, $l_8=8$.
By the definition of $sum(n)$, we obtain $sum(1)=0$, $sum(2)=\frac{1}{2}=(0.1)_2$,
$sum(3)=\frac{5}{8}=(0.101)_2$, $sum(4)=\frac{21}{32}=(0.10101)_2$, $sum(5)=\frac{43}{64}=(0.101011)_2$, $sum(6)=\frac{11}{16}=(0.1011)_2$,
$sum(7)=\frac{89}{128}=(0.1011001)_2$ and $sum(8)=\frac{45}{64}=(0.101101)_2$.
We take the first $l_n$ bits after the decimal point of the binary representation of $sum(n)$ as the $n$-th codeword;
thus, we obtain $\nu(1)=0$, $\nu(2)=100$, $\nu(3)=10100$, $\nu(4)=101010$, $\nu(5)=101011$, $\nu(6)=1011000$, $\nu(7)=1011001$ and $\nu(8)=10110100$.

From the first $8$ codewords, it can be seen that the codes obtained by the two construction schemes are not the same, but both of them satisfy Equation~\eqref{eq43}.
The advantage of the first construction scheme is its high coding efficiency; for a large integer $n$, $\nu(n)$ can be obtained relatively quickly.
The advantage of the second construction scheme is that it can be easily verified to satisfy Equation~\eqref{eq43}.}
\section{The Complete Proof of Theorem~\ref{thm5}}\label{proof}
The complete proof of Theorem~\ref{thm5} is given in this section.
We divide the proof into the following two subsections.
We first need to prove the codeword length upper bound inequalities for the $\nu$ code.

\subsection{The codeword length upper bound inequalities for the $\nu$ code}
We provide the codeword length upper bound inequalities for the $\nu$ code, as shown in the following lemma.
Notably, in our research approach, to carry out a more precise inequality scaling process for the codeword length, $F_2$ in Inequality~\eqref{eq42} is a piecewise function.
Therefore, the following lemma is divided into several cases based on the value of $P(1)$.
\begin{lemma}\label{lemma3}
Suppose that $P\in\mathcal{P}_\mathcal{A}$ is any given probability distribution.
Let
\begin{equation} \label{eq45}
g_{(c_1,c_2)}\left(a,P(1)\right)\triangleq\frac{c_1}{c_2-\log_{2}\big(1-P(1)\big)}\Big(\log_{2}a-\log_{2}\big(1-P(1)\big)\Big).
\end{equation}
\begin{enumerate}
\item When $0<P(1)\leq 1-3^{6}\cdot2^{-10}\approx 0.28809$, $L_{\nu}(a)\leq g_{(5,\log_{2}3)}(a,P(1))$ for all $3\leq a\in\mathcal{A}$;
\item when $1-3^{6}\cdot2^{-10}<P(1)\leq 0.5$, $L_{\nu}(a)\leq g_{(6,2)}(a,P(1))$ for all $3\leq a\in\mathcal{A}$;
\item when $0.5< P(1)\leq1-2^{-\frac{19}{7}}\approx0.84762$,  $L_{\nu}(a)\leq g_{(8,3)}(a,P(1))$ for all $3\leq a\in\mathcal{A}$;
\item when $1-2^{-\frac{19}{7}}< P(1) \leq 1-2^{-\frac{41}{8}}\approx0.97134$,  $L_{\nu}(a)\leq g_{(15,8)}(a,P(1))$ for all $3\leq a\in\mathcal{A}$;
\item when $1-2^{-\frac{41}{8}}<P(1)\leq1-2^{-\frac{65}{11}}\approx0.98336$, $L_{\nu}(a)\leq g_{(23,15)}(a,P(1))$ for all $3\leq a\in\mathcal{A}$;
\item when $1-2^{-\frac{65}{11}}<P(1)\leq 1-2^{-\frac{83}{13}}\approx 0.98803$, $L_{\nu}(a)\leq g_{(34,25)}(a,P(1))$ for all $3\leq a\in\mathcal{A}$;
\item when $1-2^{-\frac{83}{13}}<P(1)\leq 1-2^{-\frac{103}{15}}\approx 0.99143$,  $L_{\nu}(a)\leq g_{(47,37)}(a,P(1))$ for all $3\leq a\in\mathcal{A}$;
\item when $1-2^{-\frac{103}{15}}< P(1)\leq1-2^{-\frac{68}{9}}\approx0.99468$,  $L_{\nu}(a)\leq g_{(62,51)}(a,P(1))$ for all $3\leq a\in\mathcal{A}$;
\item when $1-2^{-\frac{68}{9}}< P(1) \leq 1-2^{-\frac{94}{11}}\approx0.99732$,  $L_{\nu}(a)\leq g_{(98,85)}(a,P(1))$ for all $3\leq a\in\mathcal{A}$; and
\item when $1-2^{-\frac{94}{11}}<P(1)\leq1-2^{-\frac{833}{65}}\approx0.99986$, $L_{\nu}(a)\leq g_{(142,127)}(a,P(1))$ for all $3\leq a\in\mathcal{A}$.
\item The length function
\[
L_{\nu}(a)\leq \frac{65}{64}\left(\frac{833}{65}+\log_{2}a\right)=\frac{833}{64}+\frac{65}{64}\log_{2}a
\]
for all $3\leq a\in\mathcal{A}$.
\end{enumerate}
\end{lemma}
\begin{proof}
\textcolor{blue}{We first outline the proof strategy.
According to the construction of Equation~\eqref{eq47}, we proceed by considering two cases.
When $3\leq a\leq7$, we directly examine the monotonicity of $\frac{c_1}{c_2-\log_{2}\big(1-P(1)\big)}\Big(\log_{2}a-\log_{2}\big(1-P(1)\big)\Big) $ with respect to $P(1)$, thereby comparing the magnitudes of $\frac{c_1}{c_2-\log_{2}\big(1-P(1)\big)}\Big(\log_{2}a-\log_{2}\big(1-P(1)\big)\Big) $ and $L_{\nu}(a)$ within the corresponding interval.
When $8\leq a\in\mathcal{A}$, $L_{\nu}(a)\leq g_{(c_1,c_2)}(a,P(1))$ is equivalent to
\begin{equation}~\label{eq96}
  1+\lfloor\log_{2}a\rfloor+2\lfloor\log_{2}(1+\lfloor\log_{2}a\rfloor)\rfloor+\Delta(a) \leq \frac{c_1}{c_2-\log_{2}\big(1-P(1)\big)}\Big(\log_{2}a-\log_{2}\big(1-P(1)\big)\Big).
\end{equation}
Let $t\triangleq \lfloor\log_{2}a\rfloor$ and $x\triangleq -\log_{2}\big(1-P(1)\big)$.
To prove~Equation \eqref{eq96}, we only need to show that
\begin{equation*}
\begin{aligned}
          & 1+t+2\lfloor\log_{2}(1+t)\rfloor+\widetilde{\Delta}(t) \leq \frac{c_1}{c_2+x}(t+x)   \\
 \iff & \left(\frac{c_1}{c_2+x}-1\right)t+\frac{c_1x}{c_2+x}-1-\widetilde{\Delta}(t)\geq 2\lfloor\log_{2}(1+t)\rfloor ,
\end{aligned}
\end{equation*}
where $\widetilde{\Delta}(\lfloor\log_{2}a\rfloor)\triangleq \Delta(a)$\footnote{According to the definition of $\Delta(\cdot)$, if $8\leq a_1\neq a_2$ and $\lfloor\log_{2}a_1\rfloor=\lfloor\log_{2}a_2\rfloor$, then $\Delta(a_1)=\Delta(a_2)$ still holds. Therefore, the definition of $\widetilde{\Delta}(\cdot)$ is well-defined, and it can also be directly defined as $\widetilde{\Delta}(t)\triangleq\Delta(2^t)$.} for all $8\leq a\in\mathcal{A}$.
Based on the definition of $\Delta(\cdot)$, $\widetilde{\Delta}(t)=0,-1,-2$ for all $3\leq t\in\mathcal{A}$.
Let
\begin{equation} \label{eq46}
  h_{(c_1,c_2)}\left(t,x\right)\triangleq\left(\frac{c_1}{c_2+x}-1\right)t+\frac{c_1x}{c_2+x}-1.
\end{equation}
Thus, to show that $L_{\nu}(a)\leq g_{(c_1,c_2)}(a,P(1))$ for all $8\leq a\in\mathcal{A}$ and $d_1 < P(1)\leq d_2 $,
it is only necessary to show that
\begin{equation}  \label{eq95}
h_{(c_1,c_2)}\left(t,x\right)-\widetilde{\Delta}(t) \geq   2\lfloor\log_{2}(1+t)\rfloor
\end{equation}
for all $3\leq t\in\mathbb{N}$ and $-\log_2(1-d_1) <  x \leq - \log_2(1-d_2)$.}
\begin{enumerate}
\item[1)] To prove that $L_{\nu}(a)\leq g_{(5,\log_{2}3)}(a,P(1))$ for all $3\leq a\in\mathcal{A}$, consider the following four cases.
\begin{enumerate}
\item When $a=3$, we obtain
    \[
      g_{(5,\log_{2}3)}(3,P(1))=\frac{5}{\log_{2}3-\log_{2}[1-P(1)]}\big\{\log_{2}3-\log_{2}[1-P(1)]\big\}=L_{\nu}(3).
    \]
\item When $a=4,5$, since $L_{\nu}(4)=L_{\nu}(5)=6$ and $g_{(5,\log_{2}3)}(a,P(1))$ is a strictly increasing sequence with respect to $a$, it is only necessary to show that $g_{(5,\log_{2}3)}(4,P(1)) \geq L_{\nu}(4)=6$. Owing to
    \[
     g_{(5,\log_{2}3)}(4,P(1))=5+\frac{10-5\log_{2}3}{\log_{2}3-\log_{2}(1-P(1))},
    \]
    $g_{(5,\log_{2}3)}(4,P(1))$ is strictly monotonically decreasing on the interval $(0,1-3^{6}\cdot2^{-10}]$ with respect to $P(1)$. Thus,
    \[
     g_{(5,\log_{2}3)}(4,P(1))\geq g_{(5,\log_{2}3)}(4,1-3^{6}\cdot2^{-10})=6=L_{\nu}(4).
    \]
\item When $a=6,7$, since $L_{\nu}(6)=L_{\nu}(7)=7$, similarly, we only need to show that $g_{(5,\log_{2}3)}(6,P(1))\geq L_{\nu}(6)=7$. Due to
    \[
     g_{(5,\log_{2}3)}(6,P(1))=5+\frac{5}{\log_{2}3-\log_{2}(1-P(1))},
    \]
    $g_{(5,\log_{2}3)}(6,P(1))$ is strictly monotonically decreasing on the interval $(0,1-3^{6}\cdot2^{-10}]$ with respect to $P(1)$. Thus,
    \[
     g_{(5,\log_{2}3)}(6,P(1))\geq g_{(5,\log_{2}3)}(6,1-3^{6}\cdot2^{-10})>7=L_{\nu}(6).
    \]
\item When $a\geq8$, from~Equation \eqref{eq95}, we only need to show that
  \begin{equation}\label{eq48}
      h_{(5,\log_{2}3)}\left(t,x\right)-\widetilde{\Delta}(t)  \geq   2\lfloor\log_{2}(1+t)\rfloor
  \end{equation}
    for all $3\leq t\in\mathbb{N}$ and $0 < x \leq 10-6\log_{2}3$. When $t=3$, the left-hand side of Equation~\eqref{eq48} is
  \begin{equation*}
\begin{aligned}
     h_{(5,\log_{2}3)}\left(3,x\right) & =  1+ \frac{15-5\log_{2}3}{x+\log_{2}3}   \\
       &\geq 1+ \frac{15-5\log_{2}3}{10-6\log_{2}3+\log_{2}3} \\
       & > 4 = 2\lfloor\log_{2}(1+3)\rfloor   .
\end{aligned}
\end{equation*}
Thereafter, each increase of 2 on the right-hand side of Equation \eqref{eq48} increases $t$ by at least $7-3=4$; thus, the left-hand side of Equation \eqref{eq48} increases by at least
\[
  \left(\frac{5}{\log_{2}3+x}-1\right)\times 4\geq \left(\frac{5}{\log_{2}3+10-6\log_{2}3}-1\right)\times 4 >4.
\]
Therefore, Equation \eqref{eq48} holds when $3\leq t\in\mathbb{N}$ and $0 < x \leq 10-6\log_{2}3$.
 \end{enumerate}
\item[2)-10)] \textcolor{blue}{The proof strategy employed for these cases is exactly the same as 1).
To avoid repetition and improve the readability of this paper, we move these proofs to Appendix~\ref{appexdix:lemma 6}.}
\item[11)]  When $3\leq a\leq255$, we obtain
\[
L_{\nu}(a)\leq L_{\nu}(255)=13 < \frac{833}{64}+\frac{65}{64}\log_{2}a.
\]
When $a\geq256$, let $t\triangleq \lfloor\log_{2}a\rfloor\geq8$; then,  $L_{\nu}(a)=L_{\delta}(a)+\Delta(a)\leq L_{\delta}(a)$.
Therefore, we only need to show that
\begin{align}\label{eq84}
          & 1+t+2\lfloor\log_{2}(1+t)\rfloor \leq \frac{833}{64}+\frac{65}{64}t  \nonumber   \\
 \iff & \frac{t}{64}+\frac{769}{64}\geq 2\lfloor\log_{2}(1+t)\rfloor .
\end{align}
We consider the following four cases.
\begin{enumerate}
\item When $8\leq t \leq 14$, the left-hand side of~Equation \eqref{eq84} is
\begin{equation*}
     \frac{t}{64}+\frac{769}{64} \geq  \frac{8}{64}+\frac{769}{64}>12> 6=2\lfloor\log_{2}(1+t)\rfloor.
\end{equation*}
\item When $15\leq t \leq 30$, the left-hand side of~Equation \eqref{eq84} is
\begin{equation*}
     \frac{t}{64}+\frac{769}{64} \geq  \frac{15}{64}+\frac{769}{64}=12.25> 8=2\lfloor\log_{2}(1+t)\rfloor.
\end{equation*}
\item When $31\leq t \leq 62$, the left-hand side of~Equation \eqref{eq84} is
\begin{equation*}
     \frac{t}{64}+\frac{769}{64} \geq  \frac{31}{64}+\frac{769}{64}=12.5> 10=2\lfloor\log_{2}(1+t)\rfloor.
\end{equation*}
\item When $63\leq t \leq 126$, the left-hand side of~Equation \eqref{eq84} is
\begin{equation*}
     \frac{t}{64}+\frac{769}{64} \geq  \frac{63}{64}+\frac{769}{64}=13> 12=2\lfloor\log_{2}(1+t)\rfloor.
\end{equation*}
\item We prove that Equation \eqref{eq84} holds for all $t\geq127$.
When $t=127$, the left-hand side of~Equation \eqref{eq84} is
\begin{equation*}
     \frac{t}{64}+\frac{769}{64}  = \frac{127}{64}+\frac{769}{64}  = 14 = 2\lfloor\log_{2}(1+t)\rfloor.
\end{equation*}
Thereafter, each increase of 2 on the right-hand side of Equation \eqref{eq84} increases $t$ by at least $255-127=128$,
and thus, the left-hand side of Equation \eqref{eq84} increases by at least $\frac{1}{64}\times 128 = 2$.
Therefore, Equation \eqref{eq84} holds when $127\leq t\in\mathbb{N}$.
\end{enumerate}
\end{enumerate}
\end{proof}
{\color{blue}\begin{remark}\label{remark2}
Lemma~\ref{lemma3} is divided into $11$ cases on the basis of the values of $P(1)$, mainly because of the joint effect of its codeword construction process and Inequality~\eqref{eq42} on the research approach.
To achieve the extension factor $C_\nu=2.0386$ of the $\nu$ code in Theorem~\ref{thm5}, the above cases cannot be merged.
However, suppose that the requirements concerning the results can be relaxed, for example, by retaining the first four cases of Lemma~\ref{lemma3} and proving that
the inequality
\[
L_{\nu}(a)\leq \frac{8}{7}\left(\frac{41}{8}+\log_{2}a\right)
\]
holds for all $3\leq a\in\mathcal{A}$.
Then, the codeword length upper bound inequality needs to be divided into only $5$ cases,
and it can subsequently be proven that the $\nu$ code has an expansion factor $C_{\nu}=\frac{15}{7}\approx2.1429$.
\end{remark}}
\subsection{The proof of Theorem~\ref{thm5}}
In this subsection, we prove the main theorem of this paper, namely, Theorem~\ref{thm5}.
To better illustrate the proof, we define some notations.
\begin{equation}\label{eq52}
\begin{aligned}
         D_{(c_1,c_2)}(P(1)) & \triangleq  \frac{c_1}{c_2-\log_{2}\big(1-P(1)\big)},  \\
         J_{(c_1,c_2)}(P(1),L) & \triangleq D_{(c_1,c_2)}(P(1))\log_{2}C_L+P(1)\Big[1+D_{(c_1,c_2)}(P(1))\big(\log_{2}P(1)-\log_{2}C_{L}\big)\Big],    \\
         R_{(c_1,c_2)}(P(1),L) &\triangleq 3+D_{(c_1,c_2)}(P(1))\left(\log_{2}\big(1-P(1)\big)+\sum_{n=2}^{L-1}\log_{2}C_{n} -(L-1)\log_{2}C_{L}\right), \\                  Q_{(c_1,c_2)}(P(1),P(2),L) & \triangleq J_{(c_1,c_2)}(P(1),L)+P(2)R_{(c_1,c_2)}(P(1),L).     \\
\end{aligned}
\end{equation}

The most important theorem of this paper is proven as follows.
\begin{theorem}(Theorem~\ref{thm5} Restated)\label{thm7}
The $\nu$ code is a UCI and has an expansion factor $C_{\nu}=2.0386$; that is,
\begin{equation*}
\frac{A_{P}(L_{\nu})}{\max\{1,H(P)\}}\leq 2.0386
\end{equation*}
for all $P\in\mathcal{P}_\mathcal{A}$.
Thus, $C^{*}\leq 2.0386$.
\end{theorem}
\begin{proof}
{\color{blue}According to Lemma~\ref{lemma3} and Lemma~\ref{lemma6}, the average codeword length of the $\nu$ code is\textcolor{red}{
\begin{equation*}
\begin{aligned}
  A_{P}(L_{\nu}) &  =  P(1)+3P(2)+ \sum_{a=3}^{\infty}P(a)L_{\nu}(a)                             \\
  &\leq P(1)+3P(2)+D_{(c_1,c_2)}(P(1))\sum_{a=3}^{\infty}P(a)\Big[\log_{2}a-\log_{2}\big(1-P(1)\big)\Big]   \\
 &\leq P(1)+3P(2)+D_{(c_1,c_2)}(P(1))\Big[H(P)+\log_{2}C_{L}+P(1)\big(\log_{2}P(1)-\log_{2}C_{L}\big)   \\
            &\quad + P(2)\Big(\log_{2}\big(1-P(1)\big)+\sum_{n=2}^{L-1}\log_{2}C_{n} -(L-1)\log_{2}C_{L}\Big) \Big ]\\
      &=D_{(c_1,c_2)}(P(1)) H(P)+D_{(c_1,c_2)}(P(1))\log_{2}C_L+P(1)\Big[1\!+\!D_{(c_1,c_2)}(P(1))\big(\log_{2}P(1)\!-\!\log_{2}C_{L}\big)\Big]   \\
            &\quad + P(2)\Big[ 3+D_{(c_1,c_2)}(P(1))\Big(\log_{2}\big(1-P(1)\big)+\sum_{n=2}^{L-1}\log_{2}C_{n}-(L-1)\log_{2}C_{L}\Big) \Big ]\\
                          & = Q_{(c_1,c_2)}(P(1),P(2),L)+D_{(c_1,c_2)}(P(1)) H(P).   \\
\end{aligned}
\end{equation*}}
Based on Equation~\eqref{eq22},
we obtain
\begin{equation} \label{eq82}
\begin{aligned}
  \frac{ A_{P}(L_{\nu})}{\max\{1,H(P)\}} & \leq \frac{ Q_{(c_1,c_2)}(P(1),P(2),L)+D_{(c_1,c_2)}(P(1)) H(P)}{\max\{1,H(P)\}}  \\
                                         & \leq Q_{(c_1,c_2)}(P(1),P(2),L)+D_{(c_1,c_2)}(P(1)).  \\
\end{aligned}
\end{equation}
Next, we analyse the strategy for further bounding the expression $Q_{(c_1,c_2)}(P(1),P(2),L)+D_{(c_1,c_2)}(P(1))$.
In expression $Q_{(c_1,c_2)}(P(1),P(2),L)+D_{(c_1,c_2)}(P(1))$, there are only two unknowns: $P(1)$ and $P(2)$. We choose to address $P(2)$.
Since
\[
 Q_{(c_1,c_2)}(P(1),P(2),L)+D_{(c_1,c_2)}(P(1))=D_{(c_1,c_2)}(P(1))+J_{(c_1,c_2)}(P(1),L)+P(2)R_{(c_1,c_2)}(P(1),L),
\]
the treatment of $P(2)$ should account for the positivity or negativity of $R_{(c_1,c_2)}(P(1),L)$.
If $R_{(c_1,c_2)}(P(1),L)<0$, then $P(2)R_{(c_1,c_2)}(P(1),L)<0$.
If $R_{(c_1,c_2)}(P(1),L)\geq0$, then
$$P(2)R_{(c_1,c_2)}(P(1),L)\leq\min\big\{P(1),1-P(1)\big\}R_{(c_1,c_2)}(P(1),L).$$
In summary, the strategy can be represented as follows:
\begin{equation*}
Q_{(c_1,c_2)}(P(1),P(2),L)\leq \left\{\begin{array}{ll}
Q_{(c_1,c_2)}(P(1),0,L),                  &\text{if } R_{(c_1,c_2)}(P(1),L)<0  \text{ ,}\\
Q_{(c_1,c_2)}(P(1),\min\{P(1),1-P(1)\},L),       &\text{if } R_{(c_1,c_2)}(P(1),L)\geq0\text{ .}\\
\end{array}\right.
\end{equation*}
After handling $P(2)$, only $P(1)$ remains as an unknown, and it can be directly addressed through differentiation.}
On the basis of the numerical value of $P(1)$, we consider the following thirteen cases.
\begin{enumerate}
\item Case $0<P(1)\leq 1-3^{6}\cdot2^{-10}\approx 0.28809$:
From 1) of Lemma~\ref{lemma3}, we know that $c_1=5$, and $c_2=\log_{2}3$.
Consider $L=2$. From Equation~\eqref{eq82}, it follows that
\begin{equation*}
\begin{aligned}
\frac{ A_{P}(L_{\nu})}{\max\{1,H(P)\}}&\leq Q_{(5,\log_{2}3)}(P(1),P(2),2)+D_{(5,\log_{2}3)}(P(1))     \\
                                               &\overset{(a)}{\leq} Q_{(5,\log_{2}3)}(P(1),P(1),2)+D_{(5,\log_{2}3)}(P(1)) ,
\end{aligned}
\end{equation*}
where $(a)$ is true because
\begin{equation*}
  R_{(c_1,c_2)}(P(1),L) = 3+D_{(5,\log_{2}3)}(P(1))\left(\log_{2}\big(1-P(1)\big)-\log_{2}\frac{3}{4}\right)>0
\end{equation*}
for all $P(1)\in(0,1-3^{6}\cdot2^{-10}]$.
Let $f_{1}(x)\triangleq Q_{(5,\log_{2}3)}(x,x,2)+D_{(5,\log_{2}3)}(x)$.
By calculating the derivative, we know that $f_{1}(x)$ is decreasing and then increasing over the interval $(0,1-3^{6}\cdot2^{-10}]$.
Thus, we obtain
\begin{equation*}
\begin{aligned}
 \frac{ A_{P}(L_{\nu})}{\max\{1,H(P)\}} & \leq f_{1}(P(1))   \\
                                                 & \leq \max\{f_{1}(0),f_{1}(1-3^{6}\cdot2^{-10})\}    \\
                                                  & = f_{1}(0)<1.8454  \\
\end{aligned}
\end{equation*}
for all $P(1)\in(0,1-3^{6}\cdot2^{-10}]$.
\item Case $1-3^{6}\cdot2^{-10}<P(1)\leq 0.5$:
From 2) of Lemma~\ref{lemma3}, we know that $c_1=6$, and $c_2=2$.
Consider $L=2$. From Equation~\eqref{eq82}, it follows that
\begin{equation*}
\begin{aligned}
\frac{ A_{P}(L_{\nu})}{\max\{1,H(P)\}}&\leq Q_{(6,2)}(P(1),P(2),2)+D_{(6,2)}(P(1))     \\
                                               &\overset{(a)}{\leq} Q_{(6,2)}(P(1),P(1),2)+D_{(6,2)}(P(1)) ,
\end{aligned}
\end{equation*}
where $(a)$ is true because
\begin{equation*}
  R_{(c_1,c_2)}(P(1),L) = 3+D_{(6,2)}(P(1))\left(\log_{2}\big(1-P(1)\big)-\log_{2}\frac{3}{4}\right)>0
\end{equation*}
for all $P(1)\in(1-3^{6}\cdot2^{-10},0.5]$.
Let \textcolor{red}{$f_{2}(x)\triangleq Q_{(6,2)}(x,x,2)+D_{(6,2)}(x)$}.
By calculating the derivative, we know that \textcolor{red}{$f_{2}(x)$} is strictly increasing over the interval $(1-3^{6}\cdot2^{-10},0.5]$.
Therefore, we obtain
\begin{equation*}
\begin{aligned}
 \frac{ A_{P}(L_{\nu})}{\max\{1,H(P)\}}  \leq f_{2}(P(1))\leq f_{2}(0.5)=2
\end{aligned}
\end{equation*}
for all $P(1)\in(1-3^{6}\cdot2^{-10},0.5]$.

\item Case $0.5< P(1)\leq 1-2^{-\frac{19}{7}}\approx 0.84762$:
From 3) of Lemma~\ref{lemma3}, we know that $c_1=8$, and $c_2=3$.
If $L=2$, then
\begin{equation*}
 R_{(c_1,c_2)}(P(1),L) = 3+D_{(8,3)}(P(1))\left(\log_{2}\big(1-P(1)\big)-\log_{2}\frac{3}{4}\right).
\end{equation*}
Let $j_1(x)\triangleq3+D_{(8,3)}(x)\left(\log_{2}\big(1-x\big)-\log_{2}\frac{3}{4}\right)$.
\textcolor{red}{By calculating the derivative}, we know that $j_{1}(x)$ is strictly decreasing over the interval $(0.5,1-2^{-\frac{19}{7}}]$, and the unique zero point over the interval is $x_1\approx0.81876$.
\begin{enumerate}
\item When $0.5<P(1)\leq x_1$, owing to Equation~\eqref{eq82}, it follows that
\begin{equation*}
\begin{aligned}
\frac{ A_{P}(L_{\nu})}{\max\{1,H(P)\}}&\leq Q_{(8,3)}(P(1),P(2),2)+D_{(8,3)}(P(1))     \\
                                               &\overset{(a)}{\leq} Q_{(8,3)}(P(1),1-P(1),2)+D_{(8,3)}(P(1)) ,
\end{aligned}
\end{equation*}
where $(a)$ is true because $R_{(8,3)}(P(1),2)=j_1(P(1))\geq0$ for all $P(1)\in(0.5,x_1]$.
Let $f_{3}(x)\triangleq Q_{(8,3)}(x,1-x,2)+D_{(8,3)}(x)$.
By calculating the derivative, we know that $f_{3}(x)$ is decreasing and then increasing over the interval $(0.5,x_1]$.
Thus, we obtain
\begin{equation*}
\begin{aligned}
 \frac{ A_{P}(L_{\nu})}{\max\{1,H(P)\}} & \leq f_{3}(P(1))   \\
                                                 & \leq \max\{f_{3}(0.5),f_{3}(x_1)\}    \\
                                                  & = f_{3}(0.5)=2 \\
\end{aligned}
\end{equation*}
for all $P(1)\in(0.5,x_1]$.
\item When $x_1<P(1)\leq 1-2^{-\frac{19}{7}}$, due to Equation~\eqref{eq82}, it follows that
\begin{equation*}
\begin{aligned}
\frac{ A_{P}(L_{\nu})}{\max\{1,H(P)\}}&\leq Q_{(8,3)}(P(1),P(2),2)+D_{(8,3)}(P(1))     \\
                                               &\overset{(a)}{\leq} Q_{(8,3)}(P(1),0,2)+D_{(8,3)}(P(1))  \\
                                               & = J_{(8,3)}(P(1),2)+D_{(8,3)}(P(1)),
\end{aligned}
\end{equation*}
where $(a)$ is true because $R_{(8,3)}(P(1),2)=j_1(P(1))<0$ for all $P(1)\in(x_1,1-2^{-\frac{19}{7}}]$.
Let $f_{4}(x)\triangleq J_{(8,3)}(x,2)+D_{(8,3)}(x)$.
By calculating the derivative, we know that $f_{4}(x)$ is strictly increasing over the interval $(x_1,1-2^{-\frac{19}{7}}]$.
Thus, we obtain
\begin{equation*}
\begin{aligned}
 \frac{ A_{P}(L_{\nu})}{\max\{1,H(P)\}}  \leq f_{4}(P(1))\leq f_{4}(1-2^{-\frac{19}{7}})<1.8761
\end{aligned}
\end{equation*}
for all $P(1)\in(x_1,1-2^{-\frac{19}{7}}]$.
\end{enumerate}

\item[4)-12)]\textcolor{blue}{The proof strategy for these cases is exactly the same as that employed for 1)-3).
To avoid repetition and improve the readability of this paper, we move these proofs to Appendix~\ref{appexdix:thm7}.}
\item[13)] Case $1-2^{-\frac{833}{65}}< P(1)< 1$:
Because $1-2^{-\frac{833}{65}}< P(1)< 1$, we obtain $\log_{2}\frac{1}{1-P(1)}>\frac{833}{65}$.
From 11) of Lemma~\ref{lemma3} and Lemma~\ref{lemma6}, the average codeword length of the $\nu$ code is
\begin{equation*}
\begin{aligned}
  A_{P}(L_{\nu}) & = P(1)+3P(2)+\sum_{i=3}^{\infty}P(i)L_{\nu}(i)    \\
                          & \leq P(1)+3P(2)+\frac{65}{64}\sum_{i=3}^{\infty}P(i)\Big(\log_{2}i+\frac{833}{65}\Big)  \\
                          & < P(1)+3P(2)+\frac{65}{64}\sum_{i=3}^{\infty}P(i)\Big[\log_{2}i-\log_{2}\big(1-P(1)\big)\Big]  \\
                          & \leq P(1)+3P(2)+\frac{65}{64}\Big[H(P)+\log_{2}\frac{3}{4}+P(1)\Big(\log_{2}P(1)-\log_{2}\frac{3}{4}\Big)    \\
                          &      \quad + P(2)\Big( \log_{2}\big(1-P(1)\big)-\log_{2}\frac{3}{4} \Big)\Big]  \\
                          & \overset{(a)}{<} \frac{65}{64}H(P)+\frac{65}{64}\log_{2}\frac{3}{4}+P(1)\Big(1+\frac{65}{64}\log_{2}P(1)-\frac{65}{64}\log_{2}\frac{3}{4}\Big),            \\
\end{aligned}
\end{equation*}
where $(a)$ is true because
\begin{equation*}
  3+\frac{65}{64}\Big( \log_{2}\big(1-P(1)\big)-\log_{2}\frac{3}{4} \Big)<0
\end{equation*}
for all $P(1)\in(1-2^{-\frac{833}{65}},1)$.
Furthermore, we have that
\begin{equation*}
\begin{aligned}
  \frac{ A_{P}(L_{\nu})}{\max\{1,H(P)\}} \leq \frac{65}{64}\Big(1+\log_{2}\frac{3}{4}\Big)+P(1)\Big(1+\frac{65}{64}\log_{2}P(1)-\frac{65}{64}\log_{2}\frac{3}{4}\Big)
\end{aligned}
\end{equation*}
due to Equation~\eqref{eq22}.
Let $f_{17}(x)\triangleq \frac{65}{64}\Big(1+\log_{2}\frac{3}{4}\Big)+x\Big(1+\frac{65}{64}\log_{2}x-\frac{65}{64}\log_{2}\frac{3}{4}\Big)$.
By calculating the derivative, we know that $f_{17}(x)$ is strictly increasing over the interval $(1-2^{-\frac{833}{65}},1)$.
Thus, we obtain
\begin{equation*}
\begin{aligned}
 \frac{ A_{P}(L_{\nu})}{\max\{1,H(P)\}} \leq f_{17}(P(1))< f_{17}(1)=2.015625
\end{aligned}
\end{equation*}
for all $P(1)\in(1-2^{-\frac{833}{65}},1)$.
\end{enumerate}
In summary, the $\nu$ code is a UCI and has an expansion factor $C_{\nu}=2.0386$. Therefore, $C^{*}\leq 2.0386$.
\end{proof}
{\color{blue}\begin{remark}
We address two points regarding the proof of Theorem~\ref{thm7} being divided into 13 cases.
First, this is because the codeword length upper bound inequalities given in Lemma~\ref{lemma3} are divided into 11 cases.
If the codeword length upper bound inequalities in Remark~\ref{remark2}, which are divided into only 5 cases, were used,
then the proof of Theorem~\ref{thm7} would only need to be divided into 6 cases.
However, the results would only show that the $\nu$ code has an expansion factor $C_{\nu}=\frac{15}{7}\approx2.1429$.
\textcolor{violet}{Second, it is important to emphasize the novelty of the obtained result.
As shown in Table \ref{tab2}, even if previously proposed codes were analysed using the proof strategy proposed in this paper,
it would still be impossible to obtain an expansion factor smaller than $2.1$.}
\end{remark}}

\section{The Lower Bound of $C_{\nu}^{*}$ and Comparisons}\label{sec_dis}
In this section, the lower bound of $C_{\nu}^{*}$ is provided, and the minimum expansion factor range of the $\nu$ code is compared with those listed in Table~\ref{tab2}.

First, we explain that $C_{\nu}^{*}>2.023936$.
Consider the probability distribution
\begin{equation*}
\overline{P}_2(a)=\left\{\begin{array}{lll}
0.992886244,            &\text{if } a=1\text{ ,}\\
\frac{1-0.992886244}{2^{132}},         &\text{if } a=2,3,\ldots,2^{132}+1 \text{ ,}\\
0,                          &\text{otherwise,}   \\
\end{array}\right.
\end{equation*}
for estimating $C_{\nu}^{*}$. We obtain
\begin{equation*}
\begin{aligned}
 \frac{ A_{\overline{P}_2}(L_{\nu})}{\max\{1,H(\overline{P}_2)\}} & = \frac{0.992886244+\frac{1-0.992886244}{2^{132}}\times\sum_{a=2}^{2^{132}+1}L_{\nu}(a)}{H(\overline{P}_2)} \\
                                             & > \frac{0.992886244+\frac{1-0.992886244}{2^{132}}\cdot7.891148088\times10^{41}}{H(\overline{P}_2)} \\
                                             & > 2.023936.
\end{aligned}
\end{equation*}
Therefore, we obtain $C_{\nu}^{*}>2.023936$.

Second, from Theorem~\ref{thm5}, we have that $2.023936<C_{\nu}^{*}\leq 2.0386$.
\textcolor{violet}{As shown in Table \ref{tab2}, the best previous result obtained regarding the upper bound of the minimum expansion factor was $C_{\iota}^{*}=2.5$, and the best previous result obtained regarding the \textcolor{red}{lower} bound of the minimum expansion factor was $C_{\omega}^{*}>2.1$.
Because $C_{\nu}^{*}\leq 2.0386<2.1$, the $\nu$ code constructed herein is currently optimal in terms of the minimum expansion factor.}

\section{Conclusions}\label{sec_con}
In this paper, we propose a new proof showing that the minimum expansion factor $C_{\mathcal{C}}^*\geq2$ for any UCI $\mathcal{C}$.
We prove a tighter probability inequality for decreasing distributions, which serves as a new tool for studying the properties of UCIs.
On the basis of this inequality, we improve the research approach for determining the upper bound of $C^*$.
A new class of UCIs, termed the $\nu$ code, is proposed. We prove that $2.023936<C_{\nu}^{*}\leq 2.0386$.
Therefore, the $\nu$ code is currently optimal in terms of the minimum expansion factor.
Since the $\nu$ code has an expansion factor $C_{\nu}=2.0386$, the upper bound of $C^*$ is reduced from $2.5$ to $2.0386$.
\textcolor{violet}{However, the proof of the upper bound of $C^*$ is not elegant because it needs to be divided into $13$ cases.}
Therefore, there are two unresolved issues, as stated below.
\begin{enumerate}
\item \textcolor{violet}{Given a specific structure or the codeword length of a class of UCIs $\mathcal{C}$, is there an algorithm that can quickly provide the exact value of the minimum expansion factor $C_{\mathcal{C}}^{*}$ or a smaller range?}
\item Although the range of $C^*$ has been narrowed to $2\leq C^*\leq2.0386$, the exact value of $C^*$ is still unknown.
\end{enumerate}

\appendices

\section{Proof of the Remaining Part of Lemma~\ref{lemma3}}
\label{appexdix:lemma 6}
\begin{proof}
\begin{enumerate}
 \item[2)] To prove that $L_{\nu}(a)\leq g_{(6,2)}(a,P(1))$ for all $3\leq a\in\mathcal{A}$, consider the following four cases.
\begin{enumerate}
\item When $a=3$, we obtain
 \begin{equation*}
\begin{aligned}
      g_{(6,2)}(3,P(1)) & =\frac{6}{2-\log_{2}[1-P(1)]}\big\{\log_{2}3-\log_{2}[1-P(1)]\big\}  \\
                        & =6+\frac{6\log_{2}3-12}{2-\log_{2}[1-P(1)]}   \\
                        & >6+\frac{6\log_{2}3-12}{2-\log_{2}(3^6\times2^{-10})}   \\
                        & = 5=L_{\nu}(3),
\end{aligned}
\end{equation*}
for all $P(1)\in(1-3^6\cdot2^{-10},0.5]$.
\item When $a=4,5$, since $L_{\nu}(4)=L_{\nu}(5)=6$ and $g_{(6,2)}(a,P(1))$ is a strictly increasing sequence with respect to $a$, it is only necessary to show that $g_{(6,2)}(4,P(1)) \geq L_{\nu}(4)=6$. We have that
    \[
     g_{(6,2)}(4,P(1))=\frac{6}{2-\log_{2}[1-P(1)]}\big\{2-\log_{2}[1-P(1)]\big\}=L_{\nu}(4).
    \]
\item When $a=6,7$, since $L_{\nu}(6)=L_{\nu}(7)=7$, similarly, we only need to show that $g_{(6,2)}(6,P(1))\geq L_{\nu}(6)=7$. We obtain
\begin{equation*}
\begin{aligned}
      g_{(6,2)}(6,P(1)) & =6+\frac{6\log_{2}6-12}{2-\log_{2}[1-P(1)]}   \\
                        & \geq 6+\frac{6\log_{2}6-12}{2-\log_{2}(1-0.5)}   \\
                        & =2+2\log_{2}6  \\
                        & >7=L_{\nu}(6),
\end{aligned}
\end{equation*}
for all $P(1)\in(1-3^6\cdot2^{-10},0.5]$.
\item When $a\geq8$, according to~Equation \eqref{eq95}, we only need to show that
  \begin{equation} \label{eq94}
      h_{(6,2)}\left(t,x\right)-\widetilde{\Delta}(t)  \geq   2\lfloor\log_{2}(1+t)\rfloor
  \end{equation}
    for all $3\leq t\in\mathbb{N}$ and $10-6\log_{2}3 < x \leq 1$. We prove that~Equation \eqref{eq94} holds when $3\leq t\in\mathbb{N}$ and $10-6\log_{2}3 < x \leq 1$ by proving that
      \begin{equation} \label{eq93}
      h_{(6,2)}\left(t,x\right) \geq   2\lfloor\log_{2}(1+t)\rfloor
  \end{equation}
    holds when $3\leq t\in\mathbb{N}$ and $10-6\log_{2}3 < x \leq 1$.

   When $t=3$, the left-hand side of~Equation \eqref{eq93} is
  \begin{equation*}
\begin{aligned}
     h_{(6,2)}\left(3,x\right) & =  \left(\frac{6}{2+x}-1\right)\times3+\frac{6x}{2+x}-1   \\
                               & =2 + \frac{6}{2+x}   \\
                               &\geq 2 + \frac{6}{2+1}  \\
                               & = 4 = 2\lfloor\log_{2}(1+3)\rfloor   .
\end{aligned}
\end{equation*}
Thereafter, each increase of 2 on the right-hand side of Equation \eqref{eq93} increases $t$ by at least $7-3=4$,
and thus, the left-hand side of Equation \eqref{eq93} increases by at least
\[
  \left(\frac{6}{2+x}-1\right)\times 4\geq \left(\frac{6}{2+1}-1\right)\times 4 =4.
\]
Therefore, Equation \eqref{eq93} holds when $3\leq t\in\mathbb{N}$ and $10-6\log_{2}3 < x \leq 1$.
 \end{enumerate}
\item[3)] To prove that $L_{\nu}(a)\leq g_{(8,3)}(a,P(1))$ for all $3\leq a\in\mathcal{A}$, consider the following four cases.
\begin{enumerate}
\item When $a=3$, we obtain
 \begin{equation*}
\begin{aligned}
      g_{(8,3)}(3,P(1)) & =8+\frac{8\log_{2}3-24}{3-\log_{2}[1-P(1)]}   \\
                        & >8+\frac{8\log_{2}3-24}{3-\log_{2}(1-0.5)}   \\
                        & > 5=L_{\nu}(3),
\end{aligned}
\end{equation*}
for all $P(1)\in(0.5,1-2^{-\frac{19}{7}}]$.
\item When $a=4,5$, since $L_{\nu}(4)=L_{\nu}(5)=6$ and $g_{(8,3)}(a,P(1))$ is a strictly increasing sequence with respect to $a$, it is only necessary to show that $g_{(8,3)}(4,P(1)) \geq L_{\nu}(4)=6$. We have that
    \begin{equation*}
\begin{aligned}
      g_{(8,3)}(4,P(1)) & =8-\frac{8}{3-\log_{2}[1-P(1)]}   \\
                        & >8-\frac{8}{3-\log_{2}(1-0.5)}   \\
                        & = 6=L_{\nu}(4),
\end{aligned}
\end{equation*}
for all $P(1)\in(0.5,1-2^{-\frac{19}{7}}]$.
\item When $a=6,7$, since $L_{\nu}(6)=L_{\nu}(7)=7$, similarly, we only need to show that $g_{(8,3)}(6,P(1))\geq L_{\nu}(6)=7$. We obtain
 \begin{equation*}
\begin{aligned}
      g_{(8,3)}(6,P(1)) & =8+\frac{8\log_{2}6-24}{3-\log_{2}[1-P(1)]}   \\
                        & >8+\frac{8\log_{2}6-24}{3-\log_{2}(1-0.5)}   \\
                        & >7=L_{\nu}(6),
\end{aligned}
\end{equation*}
for all $P(1)\in(0.5,1-2^{-\frac{19}{7}}]$.
\item When $a\geq8$, based on~Equation \eqref{eq95}, we only need to show that
  \begin{equation} \label{eq92}
      h_{(8,3)}\left(t,x\right)-\widetilde{\Delta}(t)\geq   2\lfloor\log_{2}(1+t)\rfloor
  \end{equation}
    for all $3\leq t\in\mathbb{N}$ and $1 < x \leq \frac{19}{7}$. When $3\leq t\leq 6$, the left-hand side of~Equation \eqref{eq92} is
  \begin{equation*}
\begin{aligned}
    h_{(8,3)}\left(t,x\right)-\widetilde{\Delta}(t) & =  \left(\frac{8}{3+x}-1\right)\times t+\frac{8x}{3+x}-1 - 0   \\
    & \geq  \left(\frac{8}{3+x}-1\right)\times3+\frac{8x}{3+x}-1   \\
       & = 4 = 2\lfloor\log_{2}(1+t)\rfloor   .
\end{aligned}
\end{equation*}
When $t=7$, the left-hand side of~Equation \eqref{eq92} is
  \begin{equation*}
\begin{aligned}
    h_{(8,3)}\left(7,x\right)-\widetilde{\Delta}(7) & =  \left(\frac{8}{3+x}-1\right)\times 7+\frac{8x}{3+x}-1 + 1   \\
    &  = 1+\frac{32}{3+x}\geq  1+\frac{32}{3+\frac{19}{7}}  \\
       & = 6.6 > 2\lfloor\log_{2}(1+7)\rfloor,
\end{aligned}
\end{equation*}
for all $x\in(1,\frac{19}{7}]$. When $t\geq8$, we prove that~Equation \eqref{eq92} holds by proving that
      \begin{equation} \label{eq91}
      h_{(8,3)}\left(t,x\right) \geq   2\lfloor\log_{2}(1+t)\rfloor
  \end{equation}
    holds when $8\leq t\in\mathbb{N}$ and $1< x \leq \frac{19}{7}$. When $t=8$, the left-hand side of~Equation \eqref{eq91} is
\begin{equation*}
\begin{aligned}
     h_{(8,3)}\left(8,x\right) & =  \left(\frac{8}{3+x}-1\right)\times8+\frac{8x}{3+x}-1   \\
                               & =2 + \frac{40}{3+x}-1   \\
                               &\geq  \frac{40}{3+\frac{19}{7}}-1  \\
                               & = 6 = 2\lfloor\log_{2}(1+8)\rfloor   .
\end{aligned}
\end{equation*}
Thereafter, each increase of 2 on the right-hand side of Equation \eqref{eq91} increases $t$ by at least $15-8=7$,
and thus the left-hand side of Equation \eqref{eq91} increases by at least
\[
  \left(\frac{8}{3+x}-1\right)\times 7\geq \left(\frac{8}{3+\frac{19}{7}}-1\right)\times 7 =2.8.
\]
Therefore, Equation \eqref{eq91} holds when $8\leq t\in\mathbb{N}$ and $1< x \leq \frac{19}{7}$.
 \end{enumerate}
\item[4)] When $3\leq a\leq7$, we prove that $L_{\nu}(a)\leq g_{(15,8)}(a,P(1))$ holds using the same approach as that employed in $3)$. When $8\leq a\in\mathcal{A}$, based on~Equation \eqref{eq95}, we only need to show that
  \begin{equation} \label{eq90}
      h_{(15,8)}\left(t,x\right)-\widetilde{\Delta}(t)\geq   2\lfloor\log_{2}(1+t)\rfloor
  \end{equation}
    for all $3\leq t\in\mathbb{N}$ and $\frac{19}{7} < x \leq \frac{41}{8}$.
    We consider the following seven cases.
\begin{enumerate}
\item  When $3\leq t \leq 6$, the left-hand side of~Equation \eqref{eq90} is
\begin{equation*}
\begin{aligned}
     h_{(15,8)}\left(t,x\right)-\widetilde{\Delta}(t) & =  \left(\frac{15}{8+x}-1\right)\times t+\frac{15x}{8+x}-1-0   \\
                                & \geq  \left(\frac{15}{8+x}-1\right)\times 3+\frac{15x}{8+x}-1   \\
                                & = 11-\frac{75}{8+x} > 11-\frac{75}{8+\frac{19}{7}}   \\
       & = 4 = 2\lfloor\log_{2}(1+t)\rfloor.
\end{aligned}
\end{equation*}
\item When $t=7$, the left-hand side of~Equation \eqref{eq90} is
\begin{equation*}
\begin{aligned}
     h_{(15,8)}\left(7,x\right)-\widetilde{\Delta}(7) & =  \left(\frac{15}{8+x}-1\right)\times 7+\frac{15x}{8+x}-1+1   \\
                                & = 8-\frac{15}{8+x} > 8-\frac{15}{8+\frac{19}{7}}   \\
       & = 6.6 >2\lfloor\log_{2}(1+7)\rfloor.
\end{aligned}
\end{equation*}
\item When $8\leq t \leq 14$, the left-hand side of~Equation \eqref{eq90} is
\begin{equation*}
\begin{aligned}
     h_{(15,8)}\left(t,x\right)-\widetilde{\Delta}(t) & =  \left(\frac{15}{8+x}-1\right)\times t+\frac{15x}{8+x}-1-0   \\
                                & \geq  \left(\frac{15}{8+x}-1\right)\times 8+\frac{15x}{8+x}-1   \\
       & = 6 = 2\lfloor\log_{2}(1+t)\rfloor.
\end{aligned}
\end{equation*}
\item When $15\leq t\leq 24$, the left-hand side of~Equation \eqref{eq90} is
\begin{equation*}
\begin{aligned}
     h_{(15,8)}\left(t,x\right)-\widetilde{\Delta}(t) & =  \left(\frac{15}{8+x}-1\right)\times t+\frac{15x}{8+x}-1+1   \\
                                & \geq  \left(\frac{15}{8+x}-1\right)\times 15+\frac{15x}{8+x}  \\
                                & = \frac{105}{8+x}\geq\frac{105}{8+\frac{41}{8}}  \\
       & = 8 = 2\lfloor\log_{2}(1+t)\rfloor.
\end{aligned}
\end{equation*}
\item When $25\leq t\leq 30$, the left-hand side of~Equation \eqref{eq90} is
\begin{equation*}
\begin{aligned}
     h_{(15,8)}\left(t,x\right)-\widetilde{\Delta}(t) & =  \left(\frac{15}{8+x}-1\right)\times t+\frac{15x}{8+x}-1+0   \\
                                & \geq  \left(\frac{15}{8+x}-1\right)\times 25+\frac{15x}{8+x}-1  \\
                                & = \frac{255}{8+x}-11\geq\frac{255}{8+\frac{41}{8}}-11  \\
                                & > 8 = 2\lfloor\log_{2}(1+t)\rfloor.
\end{aligned}
\end{equation*}
\item When $31\leq t\leq 40$, the left-hand side of~Equation \eqref{eq90} is
\begin{equation*}
\begin{aligned}
     h_{(15,8)}\left(t,x\right)-\widetilde{\Delta}(t) & \geq  \left(\frac{15}{8+x}-1\right)\times t+\frac{15x}{8+x}-1+1   \\
                                & \geq  \left(\frac{15}{8+x}-1\right)\times 31+\frac{15x}{8+x}  \\
                                & = \frac{345}{8+x}-16\geq\frac{345}{8+\frac{41}{8}}-16 \\
                                & > 10 = 2\lfloor\log_{2}(1+t)\rfloor.
\end{aligned}
\end{equation*}
\item When $t\geq41$, we prove that~Equation \eqref{eq90} holds by proving that
      \begin{equation} \label{eq89}
      h_{(15,8)}\left(t,x\right) \geq   2\lfloor\log_{2}(1+t)\rfloor
  \end{equation}
    holds when $41\leq t\in\mathbb{N}$ and $\frac{19}{7} < x \leq \frac{41}{8}$.
    When $t=41$, the left-hand side of~Equation \eqref{eq89} is
\begin{equation*}
\begin{aligned}
      h_{(15,8)}\left(41,x\right) & =  \left(\frac{15}{8+x}-1\right)\times 41+\frac{15x}{8+x}-1   \\
                               & = \frac{495}{8+x}-27\geq\frac{495}{8+\frac{41}{8}}-27 \\
                               & > 10 = 2\lfloor\log_{2}(1+41)\rfloor   .
\end{aligned}
\end{equation*}
Thereafter, each increase of 2 on the right-hand side of Equation \eqref{eq89} increases $t$ by at least $63-41=22$,
and thus, the left-hand side of Equation \eqref{eq89} increases by at least
\[
  \left(\frac{15}{8+x}-1\right)\times 22\geq \left(\frac{15}{8+\frac{41}{8}}-1\right)\times 22 =\frac{22}{7}.
\]
Therefore, Equation \eqref{eq89} holds when $41\leq t\in\mathbb{N}$ and $\frac{19}{7} < x \leq \frac{41}{8}$.
 \end{enumerate}

 \item[5)] When $3\leq a\leq7$, we prove that $L_{\nu}(a)\leq g_{(23,15)}(a,P(1))$ holds using the same approach as that employed for $3)$.
 When $8\leq a\in\mathcal{A}$, from~Equation \eqref{eq95}, we only need to show that
  \begin{equation} \label{eq88}
      h_{(23,15)}\left(t,x\right)-\widetilde{\Delta}(t)\geq   2\lfloor\log_{2}(1+t)\rfloor
  \end{equation}
    for all $3\leq t\in\mathbb{N}$ and $\frac{41}{8} < x \leq \frac{65}{11}$.
    We consider the following ten cases.
\begin{enumerate}
\item  When $3\leq t \leq 6$, the left-hand side of~Equation \eqref{eq88} is
\begin{equation*}
\begin{aligned}
     h_{(23,15)}\left(t,x\right)-\widetilde{\Delta}(t) & =  \left(\frac{23}{15+x}-1\right)\times t+\frac{23x}{15+x}-1-0   \\
                                & \geq  \left(\frac{23}{15+x}-1\right)\times 3+\frac{23x}{15+x}-1   \\
                                & = 19-\frac{276}{15+x} > 19-\frac{276}{15+\frac{41}{8}}   \\
       & > 5 > 2\lfloor\log_{2}(1+t)\rfloor.
\end{aligned}
\end{equation*}
\item When $t=7$, the left-hand side of~Equation \eqref{eq88} is
\begin{equation*}
\begin{aligned}
     h_{(23,15)}\left(7,x\right)-\widetilde{\Delta}(7) & =  \left(\frac{23}{15+x}-1\right)\times 7+\frac{23x}{15+x}-1+1   \\
                                & = 16-\frac{184}{15+x} > 16-\frac{184}{15+\frac{41}{8}}   \\
       & = \frac{48}{7} >2\lfloor\log_{2}(1+7)\rfloor.
\end{aligned}
\end{equation*}
\item When $8\leq t \leq 14$, the left-hand side of~Equation \eqref{eq88} is
\begin{equation*}
\begin{aligned}
     h_{(23,15)}\left(t,x\right)-\widetilde{\Delta}(t) &\geq \left(\frac{23}{15+x}-1\right)\times 8+\frac{23x}{15+x}-1   \\
        & = 14-\frac{161}{15+x} > 14-\frac{161}{15+\frac{41}{8}}   \\
       & = 6 = 2\lfloor\log_{2}(1+t)\rfloor.
\end{aligned}
\end{equation*}
\item When $15\leq t\leq 24$, the left-hand side of~Equation \eqref{eq88} is
\begin{equation*}
\begin{aligned}
     h_{(23,15)}\left(t,x\right)-\widetilde{\Delta}(t) & =  \left(\frac{23}{15+x}-1\right)\times t+\frac{23x}{15+x}-1+1   \\
                                & \geq   \left(\frac{23}{15+x}-1\right)\times 15+\frac{23x}{15+x}  \\
       & = 8 = 2\lfloor\log_{2}(1+t)\rfloor.
\end{aligned}
\end{equation*}
\item When $25\leq t\leq 30$, the left-hand side of~Equation \eqref{eq88} is
\begin{equation*}
\begin{aligned}
     h_{(23,15)}\left(t,x\right)-\widetilde{\Delta}(t) & \geq  \left(\frac{23}{15+x}-1\right)\times 25+\frac{23x}{15+x}-1  \\
                                & = \frac{230}{15+x}-3\geq\frac{230}{15+\frac{65}{11}}-3 \\
                                & = 8 = 2\lfloor\log_{2}(1+t)\rfloor.
\end{aligned}
\end{equation*}
\item When $31\leq t\leq 36$, the left-hand side of~Equation \eqref{eq88} is
\begin{equation*}
\begin{aligned}
     h_{(23,15)}\left(t,x\right)-\widetilde{\Delta}(t) & =  \left(\frac{23}{15+x}-1\right)\times t+\frac{23x}{15+x}-1+2   \\
                                & \geq  \left(\frac{23}{15+x}-1\right)\times 31+\frac{23x}{15+x}+1  \\
                                & = \frac{368}{15+x}-7\geq\frac{368}{15+\frac{65}{11}}-7 \\
                                & = 10.6 > 2\lfloor\log_{2}(1+t)\rfloor.
\end{aligned}
\end{equation*}
\item When $37\leq t\leq 50$, the left-hand side of~Equation \eqref{eq88} is
\begin{equation*}
\begin{aligned}
     h_{(23,15)}\left(t,x\right)-\widetilde{\Delta}(t) & =  \left(\frac{23}{15+x}-1\right)\times t+\frac{23x}{15+x}-1+1   \\
                                & \geq  \left(\frac{23}{15+x}-1\right)\times 37+\frac{23x}{15+x} \\
                                & = \frac{506}{15+x}-14\geq\frac{506}{15+\frac{65}{11}}-14 \\
                                & = 10.2 > 2\lfloor\log_{2}(1+t)\rfloor.
\end{aligned}
\end{equation*}
\item When $51\leq t\leq 62$, the left-hand side of~Equation \eqref{eq88} is
\begin{equation*}
\begin{aligned}
     h_{(23,15)}\left(t,x\right)-\widetilde{\Delta}(t) & =  \left(\frac{23}{15+x}-1\right)\times t+\frac{23x}{15+x}-1   \\
                                & \geq  \left(\frac{23}{15+x}-1\right)\times 51+\frac{23x}{15+x}-1 \\
                                & = \frac{828}{15+x}-29\geq\frac{828}{15+\frac{65}{11}}-29\\
                                & = 10.6 > 2\lfloor\log_{2}(1+t)\rfloor.
\end{aligned}
\end{equation*}
\item When $63\leq t\leq 84$, the left-hand side of~Equation \eqref{eq88} is
\begin{equation*}
\begin{aligned}
     h_{(23,15)}\left(t,x\right)-\widetilde{\Delta}(t) & \geq  \left(\frac{23}{15+x}-1\right)\times t+\frac{23x}{15+x}-1+1   \\
                                & \geq  \left(\frac{23}{15+x}-1\right)\times 63+\frac{23x}{15+x} \\
                                & = \frac{1104}{15+x}-40\geq\frac{1104}{15+\frac{65}{11}}-40\\
                                & = 12.8 > 2\lfloor\log_{2}(1+t)\rfloor.
\end{aligned}
\end{equation*}
\item When $t\geq 85$, we prove that~Equation \eqref{eq88} holds by proving that
      \begin{equation} \label{eq87}
      h_{(23,15)}\left(t,x\right) \geq   2\lfloor\log_{2}(1+t)\rfloor
  \end{equation}
    holds when $85\leq t\in\mathbb{N}$ and $\frac{41}{8} < x \leq \frac{65}{11}$.
    When $t=85$, the left-hand side of~Equation \eqref{eq87} is
\begin{equation*}
\begin{aligned}
      h_{(23,15)}\left(85,x\right) & =  \left(\frac{23}{15+x}-1\right)\times 85+\frac{23x}{15+x}-1   \\
                               & = \frac{1610}{15+x}-63\geq\frac{1610}{15+\frac{65}{11}}-63 \\
                               & = 14 > 2\lfloor\log_{2}(1+85)\rfloor.
\end{aligned}
\end{equation*}
Thereafter, each increase of 2 on the right-hand side of Equation \eqref{eq87} increases $t$ by at least $127-85=42$,
and thus, the left-hand side of Equation \eqref{eq87} increases by at least
\[
  \left(\frac{23}{15+x}-1\right)\times 42\geq \left(\frac{23}{15+\frac{65}{11}}-1\right)\times 42 =4.2.
\]
Therefore, Equation \eqref{eq87} holds when $85\leq t\in\mathbb{N}$ and $\frac{41}{8} < x \leq \frac{65}{11}$.
 \end{enumerate}

 \item[6)] When $3\leq a\leq7$, we prove that $L_{\nu}(a)\leq g_{(34,25)}(a,P(1))$ holds using the same approach as that employed for $3)$.
 When $8\leq a\in\mathcal{A}$, based on~Equation \eqref{eq95}, we only need to show that
  \begin{equation} \label{eq86}
      h_{(34,25)}\left(t,x\right)-\widetilde{\Delta}(t)\geq   2\lfloor\log_{2}(1+t)\rfloor
  \end{equation}
    for all $3\leq t\in\mathbb{N}$ and $\frac{65}{11} < x \leq \frac{83}{13}$.
    We consider the following nine cases.
\begin{enumerate}
\item  When $3\leq t \leq 6$, the left-hand side of~Equation \eqref{eq86} is
\begin{equation*}
\begin{aligned}
     h_{(34,25)}\left(t,x\right)-\widetilde{\Delta}(t) & =  \left(\frac{34}{25+x}-1\right)\times t+\frac{34x}{25+x}-1   \\
                                & \geq  \left(\frac{34}{25+x}-1\right)\times 3+\frac{34x}{25+x}-1   \\
                                & = 30-\frac{748}{25+x} > 30-\frac{748}{25+\frac{65}{11}}   \\
       & =5.8 > 2\lfloor\log_{2}(1+t)\rfloor.
\end{aligned}
\end{equation*}
\item When $7\leq t \leq 14$, the left-hand side of~Equation \eqref{eq86} is
\begin{equation*}
\begin{aligned}
     h_{(34,25)}\left(t,x\right)-\widetilde{\Delta}(t) & \geq   \left(\frac{34}{25+x}-1\right)\times t +\frac{34x}{25+x}-1   \\
                               & \geq     \left(\frac{34}{25+x}-1\right)\times 7+\frac{34x}{25+x}-1   \\
                               & = 26-\frac{612}{25+x} > 26-\frac{612}{25+\frac{65}{11}}   \\
       & = 6.2 > 2\lfloor\log_{2}(1+t)\rfloor.
\end{aligned}
\end{equation*}
\item When $15\leq t\leq 24$, the left-hand side of~Equation \eqref{eq86} is
\begin{equation*}
\begin{aligned}
     h_{(34,25)}\left(t,x\right)-\widetilde{\Delta}(t) & =  \left(\frac{34}{25+x}-1\right)\times t+\frac{34x}{25+x}-1+1   \\
                                & \geq   \left(\frac{34}{25+x}-1\right)\times 15+\frac{34x}{25+x}  \\
                                & = 19-\frac{340}{25+x} > 19-\frac{340}{25+\frac{65}{11}}   \\
       & = 8 = 2\lfloor\log_{2}(1+t)\rfloor.
\end{aligned}
\end{equation*}
\item When $25\leq t\leq 30$, the left-hand side of~Equation \eqref{eq86} is
\begin{equation*}
\begin{aligned}
     h_{(34,25)}\left(t,x\right)-\widetilde{\Delta}(t) & =  \left(\frac{34}{25+x}-1\right)\times t+\frac{34x}{25+x}-1  \\
                                   &   \geq  \left(\frac{34}{25+x}-1\right)\times 25+\frac{34x}{25+x}-1  \\
                                & = 8 = 2\lfloor\log_{2}(1+t)\rfloor.
\end{aligned}
\end{equation*}
\item When $31\leq t\leq 36$, the left-hand side of~Equation \eqref{eq86} is
\begin{equation*}
\begin{aligned}
     h_{(34,25)}\left(t,x\right)-\widetilde{\Delta}(t) & =  \left(\frac{34}{25+x}-1\right)\times t+\frac{34x}{25+x}-1+2   \\
                                & \geq  \left(\frac{34}{25+x}-1\right)\times 31+\frac{34x}{25+x}+1  \\
                                & = \frac{204}{25+x}+4\geq\frac{204}{25+\frac{83}{13}}+4 \\
                                & = 10.5 > 2\lfloor\log_{2}(1+t)\rfloor.
\end{aligned}
\end{equation*}
\item When $37\leq t\leq 50$, the left-hand side of~Equation \eqref{eq86} is
\begin{equation*}
\begin{aligned}
     h_{(34,25)}\left(t,x\right)-\widetilde{\Delta}(t) & =  \left(\frac{34}{25+x}-1\right)\times t+\frac{34x}{25+x}-1+1   \\
                                & \geq  \left(\frac{34}{25+x}-1\right)\times 37+\frac{34x}{25+x} \\
                                & = \frac{408}{25+x}-3\geq\frac{408}{25+\frac{83}{13}}-3\\
                                & = 10 = 2\lfloor\log_{2}(1+t)\rfloor.
\end{aligned}
\end{equation*}
\item When $51\leq t\leq 62$, the left-hand side of~Equation \eqref{eq86} is
\begin{equation*}
\begin{aligned}
     h_{(34,25)}\left(t,x\right)-\widetilde{\Delta}(t) & \geq  \left(\frac{34}{25+x}-1\right)\times 51+\frac{34x}{25+x}-1 \\
                                & = \frac{884}{25+x}-18\geq\frac{884}{25+\frac{83}{13}}-18\\
                                & > 10 = 2\lfloor\log_{2}(1+t)\rfloor.
\end{aligned}
\end{equation*}
\item When $63\leq t\leq 84$, the left-hand side of~Equation \eqref{eq86} is
\begin{equation*}
\begin{aligned}
     h_{(34,25)}\left(t,x\right)-\widetilde{\Delta}(t) & \geq  \left(\frac{34}{25+x}-1\right)\times t+\frac{34x}{25+x}-1+1   \\
                                & \geq  \left(\frac{34}{25+x}-1\right)\times 63+\frac{34x}{25+x} \\
                                & = \frac{1292}{25+x}-29\geq\frac{1292}{25+\frac{83}{13}}-29\\
                                & > 12 = 2\lfloor\log_{2}(1+t)\rfloor.
\end{aligned}
\end{equation*}
\item When $t\geq 85$, we prove that~Equation \eqref{eq86} holds by proving that
      \begin{equation} \label{eq85}
      h_{(34,25)}\left(t,x\right) \geq   2\lfloor\log_{2}(1+t)\rfloor
  \end{equation}
    holds when $85\leq t\in\mathbb{N}$ and $\frac{65}{11} < x \leq \frac{83}{13}$.
    When $t=85$, the left-hand side of~Equation \eqref{eq85} is
\begin{equation*}
\begin{aligned}
      h_{(34,25)}\left(85,x\right) & =  \left(\frac{34}{25+x}-1\right)\times 85+\frac{34x}{25+x}-1   \\
                               & = \frac{2040}{25+x}-52\geq\frac{2040}{25+\frac{83}{13}}-52 \\
                               & = 13 > 2\lfloor\log_{2}(1+85)\rfloor.
\end{aligned}
\end{equation*}
Thereafter, each increase of 2 on the right-hand side of Equation \eqref{eq85} increases $t$ by at least $127-85=42$,
and thus, the left-hand side of Equation \eqref{eq85} increases by at least
\[
  \left(\frac{34}{25+x}-1\right)\times 42\geq \left(\frac{34}{25+\frac{83}{13}}-1\right)\times 42 =3.5.
\]
Therefore, Equation \eqref{eq85} holds when $85\leq t\in\mathbb{N}$ and $\frac{65}{11} < x \leq \frac{83}{13}$.
 \end{enumerate}

\item[7)] When $3\leq a\leq7$, we prove that $L_{\nu}(a)\leq g_{(47,37)}(a,P(1))$ holds using the same way as that employed for $3)$.
When $8\leq a\in\mathcal{A}$, we prove that $L_{\nu}(a)\leq g_{(47,37)}(a,P(1))$ holds using the same approach as that employed for $5)$ and $6)$.

\item[8)] When $3\leq a\leq7$, we prove that $L_{\nu}(a)\leq g_{(62,51)}(a,P(1))$ holds using the same approach as that employed for $3)$.
When $8\leq a\in\mathcal{A}$, we prove that $L_{\nu}(a)\leq g_{(62,51)}(a,P(1))$ holds using the same approach as that employed for $5)$ and $6)$.

\item[9)] When $3\leq a\leq7$, we prove that $L_{\nu}(a)\leq g_{(98,85)}(a,P(1))$ holds using the same approach as that employed for $3)$.
When $8\leq a\in\mathcal{A}$, we prove that $L_{\nu}(a)\leq g_{(98,85)}(a,P(1))$ holds using the same approach as that employed for $5)$ and $6)$.

\item[10)] When $3\leq a\leq7$, we prove that $L_{\nu}(a)\leq g_{(142,127)}(a,P(1))$ holds using the same approach as that employed for $3)$.
When $8\leq a\in\mathcal{A}$, we prove that $L_{\nu}(a)\leq g_{(142,127)}(a,P(1))$ holds using the same approach as that employed for $5)$ and $6)$.
\end{enumerate}
\end{proof}
\section{Proof of the Remaining Part of Theorem~\ref{thm7} }\label{appexdix:thm7}
\begin{proof}
\begin{enumerate}
\item[4)] Case $1-2^{-\frac{19}{7}}<P(1)\leq 0.92$: From 4) of Lemma~\ref{lemma3}, we know that $c_1=15$, and $c_2=8$.
Consider $L=4$. From Equation~\eqref{eq82}, it follows that
\begin{equation*}
\begin{aligned}
\frac{ A_{P}(L_{\nu})}{\max\{1,H(P)\}}&\leq Q_{(15,8)}(P(1),P(2),4)+D_{(15,8)}(P(1))     \\
                                               &\overset{(a)}{\leq} Q_{(15,8)}(P(1),1-P(1),4)+D_{(15,8)}(P(1)) ,
\end{aligned}
\end{equation*}
where $(a)$ is true because
\begin{equation*}
  R_{(c_1,c_2)}(P(1),L) =  3+D_{(15,8)}(P(1))\left(\log_{2}\big(1-P(1)\big)+\sum_{n=2}^{3}\log_{2}C_{n} -3\log_{2}C_{4}\right)>0
\end{equation*}
for all $P(1)\in(1-2^{-\frac{19}{7}},0.92]$.
Let $f_{5}(x)\triangleq Q_{(15,8)}(x,1-x,4)+D_{(15,8)}(x)$.
By calculating the derivative, we know that $f_{5}(x)$ is decreasing and then increasing over the interval $(1-2^{-\frac{19}{7}},0.92]$.
Therefore, we obtain
\begin{equation*}
\begin{aligned}
 \frac{ A_{P}(L_{\nu})}{\max\{1,H(P)\}} & \leq f_{5}(P(1))   \\
                                                 & \leq \max\{f_{5}(1-2^{-\frac{19}{7}}),f_{5}(0.92)\}    \\
                                                  & = f_{5}(0.92)<1.9999 \\
\end{aligned}
\end{equation*}
for all $P(1)\in(1-2^{-\frac{19}{7}},0.92]$.

\item[5)] Case $0.92<P(1)\leq 1-2^{-\frac{41}{8}}\approx0.97134$: From 4) of Lemma~\ref{lemma3}, we know that $c_1=15$, and $c_2=8$.
If $L=5$, then
\begin{equation*}
 R_{(c_1,c_2)}(P(1),L) = 3+D_{(15,8)}(P(1))\left(\log_{2}\big(1-P(1)\big)+\sum_{n=2}^{4}\log_{2}C_{n} -4\log_{2}C_{5}\right).
\end{equation*}
Let $j_2(x)\triangleq 3+D_{(15,8)}(x)\left(\log_{2}\big(1-x\big)+\sum_{n=2}^{4}\log_{2}C_{n} -4\log_{2}C_{5}\right)$.
\textcolor{red}{By calculating the derivative}, we know that $j_{2}(x)$ is strictly decreasing over the interval $(0.92,1-2^{-\frac{41}{8}}]$, and the unique zero point over the interval is $x_2\approx0.95602$.
\begin{enumerate}
\item When $0.92<P(1)\leq x_2$, owing to Equation~\eqref{eq82}, it follows that
\begin{equation*}
\begin{aligned}
\frac{ A_{P}(L_{\nu})}{\max\{1,H(P)\}}&\leq Q_{(15,8)}(P(1),P(2),5)+D_{(15,8)}(P(1))     \\
                                               &\overset{(a)}{\leq} Q_{(15,8)}(P(1),1-P(1),5)+D_{(15,8)}(P(1)) ,
\end{aligned}
\end{equation*}
where $(a)$ is true because  $R_{(15,8)}(P(1),5)=j_2(P(1))\geq0$ for all $P(1)\in(0.92,x_2]$.
Let $f_{6}(x)\triangleq Q_{(15,8)}(x,1-x,5)+D_{(15,8)}(x)$.
By calculating the derivative, we know that $f_{6}(x)$ is strictly decreasing over the interval $(0.92,x_2]$.
Thus, we obtain
\begin{equation*}
\begin{aligned}
 \frac{ A_{P}(L_{\nu})}{\max\{1,H(P)\}}\leq f_{6}(P(1)) < f_{6}(0.92)<2.0313  \\
\end{aligned}
\end{equation*}
for all $P(1)\in(0.92,x_2]$.
\item When $x_2<P(1)\leq 1-2^{-\frac{41}{8}}$, according to Equation~\eqref{eq82}, it follows that
\begin{equation*}
\begin{aligned}
\frac{ A_{P}(L_{\nu})}{\max\{1,H(P)\}}&\leq Q_{(15,8)}(P(1),P(2),5)+D_{(15,8)}(P(1))      \\
                                               &\overset{(a)}{\leq} J_{(15,8)}(P(1),5)+D_{(15,8)}(P(1)),
\end{aligned}
\end{equation*}
where $(a)$ is true because  $R_{(15,8)}(P(1),5)=j_2(P(1))<0$ for all $P(1)\in(x_2,1-2^{-\frac{41}{8}}]$.
Let $f_{7}(x)\triangleq J_{(15,8)}(x,5)+D_{(15,8)}(x)$.
By calculating the derivative, we know that $f_{7}(x)$ is increasing and then decreasing over the interval $(x_2,1-2^{-\frac{41}{8}}]$,
and $f_{7}(x)$ takes its maximum value over $(x_2,1-2^{-\frac{41}{8}}]$ at $x_3\approx0.96883$.
Therefore, we obtain
\begin{equation*}
\begin{aligned}
 \frac{ A_{P}(L_{\nu})}{\max\{1,H(P)\}}  \leq f_{7}(P(1))\leq f_{7}(x_3)<2.0345
\end{aligned}
\end{equation*}
for all $P(1)\in(x_2,1-2^{-\frac{41}{8}}]$.
\end{enumerate}

\item[6)] Case $1-2^{-\frac{41}{8}}<P(1)\leq 1-2^{-\frac{65}{11}}\approx0.98336$: From 5) of Lemma~\ref{lemma3}, we know that $c_1=23$, and $c_2=15$.
Consider $L=6$. From Equation~\eqref{eq82}, it follows that
\begin{equation*}
\begin{aligned}
\frac{ A_{P}(L_{\nu})}{\max\{1,H(P)\}}&\leq Q_{(23,15)}(P(1),P(2),6)+D_{(23,15)}(P(1))     \\
                                               &\overset{(a)}{\leq} J_{(23,15)}(P(1),6)+D_{(23,15)}(P(1)) ,
\end{aligned}
\end{equation*}
where $(a)$ is true because
\begin{equation*}
  R_{(c_1,c_2)}(P(1),L) =  3+D_{(23,15)}(P(1))\left(\log_{2}\big(1-P(1)\big)+\sum_{n=2}^{5}\log_{2}C_{n} -5\log_{2}C_{6}\right)<0
\end{equation*}
for all $P(1)\in(1-2^{-\frac{41}{8}},1-2^{-\frac{65}{11}}]$.
Let $f_{8}(x)\triangleq J_{(23,15)}(x,6)+D_{(23,15)}(x)$.
By calculating the derivative, we know that $f_{8}(x)$ is increasing and then decreasing over the interval $(1-2^{-\frac{41}{8}},1-2^{-\frac{65}{11}}]$,
and $f_{8}(x)$ takes its maximum value over $(1-2^{-\frac{41}{8}},1-2^{-\frac{65}{11}}]$ at $x_4\approx0.98053$.
Therefore, we obtain
\begin{equation*}
\begin{aligned}
 \frac{ A_{P}(L_{\nu})}{\max\{1,H(P)\}}  \leq f_{8}(P(1))\leq f_{8}(x_4)<2.0380
\end{aligned}
\end{equation*}
for all $P(1)\in(1-2^{-\frac{41}{8}},1-2^{-\frac{65}{11}}]$.

\item[7)] Case $1-2^{-\frac{65}{11}}<P(1)\leq 1-2^{-\frac{83}{13}}\approx0.98803$: From 6) of Lemma~\ref{lemma3}, we know that $c_1=34$, and $c_2=25$.
Consider $L=8$. From Equation~\eqref{eq82}, it follows that
\begin{equation*}
\begin{aligned}
\frac{ A_{P}(L_{\nu})}{\max\{1,H(P)\}}&\leq Q_{(34,25)}(P(1),P(2),8)+D_{(34,25)}(P(1))     \\
                                               &\overset{(a)}{\leq} J_{(34,25)}(P(1),8)+D_{(34,25)}(P(1)) ,
\end{aligned}
\end{equation*}
where $(a)$ is true because
\begin{equation*}
  R_{(c_1,c_2)}(P(1),L) =  3+D_{(34,25)}(P(1))\left(\log_{2}\big(1-P(1)\big)+\sum_{n=2}^{7}\log_{2}C_{n} -7\log_{2}C_{8}\right)<0
\end{equation*}
for all $P(1)\in(1-2^{-\frac{65}{11}},1-2^{-\frac{83}{13}}]$.
Let $f_{9}(x)\triangleq J_{(34,25)}(x,8)+D_{(34,25)}(x)$.
By calculating the derivative, we know that $f_{9}(x)$ is increasing and then decreasing over the interval $(1-2^{-\frac{65}{11}},1-2^{-\frac{83}{13}}]$,
and $f_{9}(x)$ takes its maximum value over $(1-2^{-\frac{65}{11}},1-2^{-\frac{83}{13}}]$ at $x_5\approx0.98735$.
Therefore, we obtain
\begin{equation*}
\begin{aligned}
 \frac{ A_{P}(L_{\nu})}{\max\{1,H(P)\}}\leq f_{9}(P(1))\leq f_{9}(x_5)<2.0376
\end{aligned}
\end{equation*}
for all $P(1)\in(1-2^{-\frac{65}{11}},1-2^{-\frac{83}{13}}]$.

\item[8)] Case $1-2^{-\frac{83}{13}}<P(1)\leq 1-2^{-\frac{103}{15}}\approx0.99143$: From 7) of Lemma~\ref{lemma3}, we know that $c_1=47$, and $c_2=37$.
If $L=12$, then
\begin{equation*}
 R_{(c_1,c_2)}(P(1),L) = 3+D_{(47,37)}(P(1))\left(\log_{2}\big(1-P(1)\big)+\sum_{n=2}^{11}\log_{2}C_{n} -11\log_{2}C_{12}\right).
\end{equation*}
Let $j_3(x)\triangleq 3+D_{(47,37)}(x)\left(\log_{2}\big(1-x\big)+\sum_{n=2}^{11}\log_{2}C_{n} -11\log_{2}C_{12}\right)$.
\textcolor{red}{By calculating the derivative}, we know that $j_{3}(x)$ is strictly decreasing over the interval $(1-2^{-\frac{83}{13}},1-2^{-\frac{103}{15}}]$, and the unique zero point over this interval is $x_{6}\approx0.98876$.
\begin{enumerate}
\item When $1-2^{-\frac{83}{13}}<P(1)\leq x_{6}$, based on Equation~\eqref{eq82}, it follows that
\begin{equation*}
\begin{aligned}
\frac{ A_{P}(L_{\nu})}{\max\{1,H(P)\}}&\leq Q_{(47,37)}(P(1),P(2),12)+D_{(47,37)}(P(1))     \\
                                               &\overset{(a)}{\leq} Q_{(47,37)}(P(1),1-P(1),12)+D_{(47,37)}(P(1)) ,
\end{aligned}
\end{equation*}
where $(a)$ is true because  $R_{(47,37)}(P(1),12)=j_3(P(1))\geq0$ for all $P(1)\in(1-2^{-\frac{83}{13}},x_{6}]$.
Let $f_{10}(x)\triangleq Q_{(47,37)}(x,1-x,12)+D_{(47,37)}(x)$.
By calculating the derivative, we know that $f_{10}(x)$ is strictly decreasing over the interval $(1-2^{-\frac{83}{13}},x_{6}]$.
Thus, we obtain
\begin{equation*}
\begin{aligned}
 \frac{ A_{P}(L_{\nu})}{\max\{1,H(P)\}}\leq f_{10}(P(1)) < f_{10}(1-2^{-\frac{83}{13}})<2.0375  \\
\end{aligned}
\end{equation*}
for all $P(1)\in(1-2^{-\frac{83}{13}},x_{6}]$.
\item When $x_{6}<P(1)\leq 1-2^{-\frac{103}{15}}$, owing to Equation~\eqref{eq82}, it follows that
\begin{equation*}
\begin{aligned}
\frac{ A_{P}(L_{\nu})}{\max\{1,H(P)\}}&\leq Q_{(47,37)}(P(1),P(2),12)+D_{(47,37)}(P(1))      \\
                                               &\overset{(a)}{\leq} J_{(47,37)}(P(1),12)+D_{(47,37)}(P(1)),
\end{aligned}
\end{equation*}
where $(a)$ is true because $R_{(47,37)}(P(1),12)=j_3(P(1))<0$ for all $P(1)\in(x_{6},1-2^{-\frac{103}{15}}]$.
Let $f_{11}(x)\triangleq J_{(47,37)}(x,12)+D_{(47,37)}(x)$.
By calculating the derivative, we know that $f_{11}(x)$ is increasing and then decreasing over the interval $(x_{6},1-2^{-\frac{103}{15}}]$,
and $f_{11}(x)$ takes its maximum value over $(x_{6},1-2^{-\frac{103}{15}}]$ at $x_{7}\approx0.99114$.
Therefore, we obtain
\begin{equation*}
\begin{aligned}
 \frac{ A_{P}(L_{\nu})}{\max\{1,H(P)\}}  \leq f_{11}(P(1))\leq f_{11}(x_{7})<2.0381
\end{aligned}
\end{equation*}
for all $P(1)\in(x_{6},1-2^{-\frac{103}{15}}]$.
\end{enumerate}

\item[9)] Case $1-2^{-\frac{103}{15}}<P(1)\leq 0.9927$: From 8) of Lemma~\ref{lemma3}, we know that $c_1=62$, and $c_2=51$.
If $L=15$, then
\begin{equation*}
 R_{(c_1,c_2)}(P(1),L) = 3+D_{(62,51)}(P(1))\left(\log_{2}\big(1-P(1)\big)+\sum_{n=2}^{14}\log_{2}C_{n} -14\log_{2}C_{15}\right).
\end{equation*}
Let $j_4(x)\triangleq 3+D_{(62,51)}(x)\left(\log_{2}\big(1-x\big)+\sum_{n=2}^{14}\log_{2}C_{n} -14\log_{2}C_{15}\right)$.
\textcolor{red}{By calculating the derivative}, we know that $j_{4}(x)$ is strictly decreasing over the interval $(1-2^{-\frac{103}{15}},0.9927]$, and the unique zero point over the interval is $x_{8}\approx0.99195$.
\begin{enumerate}
\item When $1-2^{-\frac{103}{15}}<P(1)\leq x_{8}$, due to Equation~\eqref{eq82}, it follows that
\begin{equation*}
\begin{aligned}
\frac{ A_{P}(L_{\nu})}{\max\{1,H(P)\}}&\leq Q_{(62,51)}(P(1),P(2),15)+D_{(62,51)}(P(1))     \\
                                               &\overset{(a)}{\leq} Q_{(62,51)}(P(1),1-P(1),15)+D_{(62,51)}(P(1)) ,
\end{aligned}
\end{equation*}
where $(a)$ is true because $R_{(62,51)}(P(1),15)=j_4(P(1))\geq0$ for all $P(1)\in(1-2^{-\frac{103}{15}},x_{8}]$.
Let $f_{12}(x)\triangleq Q_{(62,51)}(x,1-x,15)+D_{(62,51)}(x)$.
By calculating the derivative, we know that $f_{12}(x)$ is strictly decreasing over the interval $(1-2^{-\frac{103}{15}},x_{8}]$.
Thus, we obtain
\begin{equation*}
\begin{aligned}
 \frac{ A_{P}(L_{\nu})}{\max\{1,H(P)\}}\leq f_{12}(P(1)) < f_{12}(1-2^{-\frac{103}{15}})<2.0386  \\
\end{aligned}
\end{equation*}
for all $P(1)\in(1-2^{-\frac{103}{15}},x_{8}]$.
\item When $x_{8}<P(1)\leq 0.9927$, because of Equation~\eqref{eq82}, it follows that
\begin{equation*}
\begin{aligned}
\frac{ A_{P}(L_{\nu})}{\max\{1,H(P)\}}&\leq Q_{(62,51)}(P(1),P(2),15)+D_{(62,51)}(P(1))     \\
                                               &\overset{(a)}{\leq} J_{(62,51)}(P(1),15)+D_{(62,51)}(P(1)),
\end{aligned}
\end{equation*}
where $(a)$ is true because $R_{(62,51)}(P(1),15)=j_4(P(1))<0$ for all $P(1)\in(x_{8},0.9927]$.
Let $f_{13}(x)\triangleq J_{(62,51)}(x,15)+D_{(62,51)}(x)$.
By calculating the derivative, we know that $f_{13}(x)$ is strictly increasing over the interval $(x_{8},0.9927]$.
Therefore, we obtain
\begin{equation*}
\begin{aligned}
 \frac{ A_{P}(L_{\nu})}{\max\{1,H(P)\}}  \leq f_{13}(P(1))\leq f_{13}(0.9927)<2.0386
\end{aligned}
\end{equation*}
for all $P(1)\in(x_{8},0.9927]$.
\end{enumerate}

\item[10)] Case $0.9927<P(1)\leq 1-2^{-\frac{68}{9}} \approx0.99468$:
From 8) of Lemma~\ref{lemma3}, we know that $c_1=62$, and $c_2=51$.
Consider $L=16$. From Equation~\eqref{eq82}, it follows that
\begin{equation*}
\begin{aligned}
\frac{ A_{P}(L_{\nu})}{\max\{1,H(P)\}}&\leq Q_{(62,51)}(P(1),P(2),16)+D_{(62,51)}(P(1))      \\
                                               &\overset{(a)}{\leq} J_{(62,51)}(P(1),16)+D_{(62,51)}(P(1)),
\end{aligned}
\end{equation*}
where $(a)$ is true because
\begin{equation*}
  R_{(c_1,c_2)}(P(1),L)  = 3+D_{(62,51)}(P(1))\left(\log_{2}\big(1-P(1)\big)+\sum_{n=2}^{15}\log_{2}C_{n} -15\log_{2}C_{16}\right)<0
\end{equation*}
for all $P(1)\in(0.9927,1-2^{-\frac{68}{9}}]$.
Let $f_{14}(x)\triangleq J_{(62,51)}(x,16)+D_{(62,51)}(x)$.
By calculating the derivative, we know that $f_{14}(x)$ is increasing and then decreasing over the interval $(0.9927,1-2^{-\frac{68}{9}}]$,
and $f_{14}(x)$ takes its maximum value over $(0.9927,1-2^{-\frac{68}{9}}]$ at $x_{9}\approx0.99339$.
Therefore, we obtain
\begin{equation*}
\begin{aligned}
 \frac{ A_{P}(L_{\nu})}{\max\{1,H(P)\}} \leq f_{14}(P(1))\leq f_{14}(x_{9})<2.0386
\end{aligned}
\end{equation*}
for all $P(1)\in(0.9927,1-2^{-\frac{68}{9}}]$.

\item[11)] Case $1-2^{-\frac{68}{9}}<P(1)\leq 1-2^{-\frac{94}{11}} \approx0.99732$:
From 9) of Lemma~\ref{lemma3}, we know that $c_1=98$, and $c_2=85$.
Consider $L=18$. From Equation~\eqref{eq82}, it follows that
\begin{equation*}
\begin{aligned}
\frac{ A_{P}(L_{\nu})}{\max\{1,H(P)\}}&\leq Q_{(98,85)}(P(1),P(2),18)+D_{(98,85)}(P(1))      \\
                                     &\overset{(a)}{\leq} J_{(98,85)}(P(1),18)+D_{(98,85)}(P(1)),
\end{aligned}
\end{equation*}
where $(a)$ is true because
\begin{equation*}
  R_{(c_1,c_2)}(P(1),L)  = 3+D_{(98,85)}(P(1))\left(\log_{2}\big(1-P(1)\big)+\sum_{n=2}^{17}\log_{2}C_{n} -17\log_{2}C_{18}\right)<0
\end{equation*}
for all $P(1)\in(1-2^{-\frac{68}{9}},1-2^{-\frac{94}{11}}]$.
Let $f_{15}(x)\triangleq J_{(98,85)}(x,18)+D_{(98,85)}(x)$.
By calculating the derivative, we know that $f_{15}(x)$ is increasing and then decreasing over the interval $(1-2^{-\frac{68}{9}},1-2^{-\frac{94}{11}}]$,
and $f_{15}(x)$ takes its maximum value over $(1-2^{-\frac{68}{9}},1-2^{-\frac{94}{11}}]$ at $x_{10}\approx0.99586$.
Therefore, we obtain
\begin{equation*}
\begin{aligned}
 \frac{ A_{P}(L_{\nu})}{\max\{1,H(P)\}} \leq f_{15}(P(1))\leq f_{15}(x_{10})<2.0386
\end{aligned}
\end{equation*}
for all $P(1)\in(1-2^{-\frac{68}{9}},1-2^{-\frac{94}{11}}]$.

\item[12)] Case $1-2^{-\frac{94}{11}}<P(1)\leq 1-2^{-\frac{833}{65}} \approx0.99986$:
From 10) of Lemma~\ref{lemma3}, we know that $c_1=142$, and $c_2=127$.
Consider $L=18$. From Equation~\eqref{eq82}, it follows that
\begin{equation*}
\begin{aligned}
\frac{ A_{P}(L_{\nu})}{\max\{1,H(P)\}}&\leq Q_{(142,127)}(P(1),P(2),18)+D_{(142,127)}(P(1))      \\
                                     &\overset{(a)}{\leq} J_{(142,127)}(P(1),18)+D_{(142,127)}(P(1)),
\end{aligned}
\end{equation*}
where $(a)$ is true because
\begin{equation*}
  R_{(c_1,c_2)}(P(1),L)  = 3+D_{(142,127)}(P(1))\left(\log_{2}\big(1-P(1)\big)+\sum_{n=2}^{17}\log_{2}C_{n} -17\log_{2}C_{18}\right)<0
\end{equation*}
for all $P(1)\in(1-2^{-\frac{94}{11}},1-2^{-\frac{833}{65}}]$.
Let $f_{16}(x)\triangleq J_{(142,127)}(x,18)+D_{(142,127)}(x)$.
By calculating the derivative, we know that $f_{16}(x)$ is strictly decreasing over the interval $(1-2^{-\frac{94}{11}},1-2^{-\frac{833}{65}}]$.
Therefore, we obtain
\begin{equation*}
\begin{aligned}
 \frac{ A_{P}(L_{\nu})}{\max\{1,H(P)\}} \leq f_{16}(P(1))< f_{16}(1-2^{-\frac{94}{11}})<2.0372
\end{aligned}
\end{equation*}
for all $P(1)\in(1-2^{-\frac{94}{11}},1-2^{-\frac{833}{65}}]$.
\end{enumerate}
\end{proof}
\bibliographystyle{IEEEtran}
\bibliography{IEEEabrv,refs}

\begin{thebibliography}{10}
\providecommand{\url}[1]{#1}
\csname url@samestyle\endcsname
\providecommand{\newblock}{\relax}
\providecommand{\bibinfo}[2]{#2}
\providecommand{\BIBentrySTDinterwordspacing}{\spaceskip=0pt\relax}
\providecommand{\BIBentryALTinterwordstretchfactor}{4}
\providecommand{\BIBentryALTinterwordspacing}{\spaceskip=\fontdimen2\font plus
\BIBentryALTinterwordstretchfactor\fontdimen3\font minus
  \fontdimen4\font\relax}
\providecommand{\BIBforeignlanguage}[2]{{%
\expandafter\ifx\csname l@#1\endcsname\relax
\typeout{** WARNING: IEEEtran.bst: No hyphenation pattern has been}%
\typeout{** loaded for the language `#1'. Using the pattern for}%
\typeout{** the default language instead.}%
\else
\language=\csname l@#1\endcsname
\fi
#2}}
\providecommand{\BIBdecl}{\relax}
\BIBdecl

\bibitem{SH481}
C.~E. {Shannon}, ``A mathematical theory of communication,'' \emph{The Bell
  System Technical Journal}, vol.~27, no.~3, pp. 379--423, Jul. 1948.

\bibitem{SH482}
C.~E. {Shannon}, ``A mathematical theory of communication,'' \emph{The Bell
  System Technical Journal}, vol.~27, no.~4, pp. 623--656, Oct. 1948.

\bibitem{H52}
D.~A. {Huffman}, ``A method for the construction of minimum-redundancy codes,''
  \emph{Proceedings of the IRE}, vol.~40, no.~9, pp. 1098--1101, Sep. 1952.

\bibitem{P76}
R.~C. {Pasco}, \emph{Source Coding Algorithms for Fast Data Compression}.\hskip
  1em plus 0.5em minus 0.4em\relax Ph.D. Dissertation, Stanford University, CA,
  USA, May 1976.

\bibitem{AC79}
J.~{Rissanen} and G.~G. {Langdon}, ``Arithmetic coding,'' \emph{IBM Journal of
  Research and Development}, vol.~23, no.~2, pp. 149--162, Mar. 1979.

\bibitem{Elias75}
P.~Elias, ``Universal codeword sets and representations of the integers,''
  \emph{IEEE Transactions on Information Theory}, vol.~21, no.~2, pp. 194--203,
  Mar. 1975.

\bibitem{N19}
U.~{Niesen}, ``An information-theoretic analysis of deduplication,'' \emph{IEEE
  Transactions on Information Theory}, vol.~65, no.~9, pp. 5688--5704, Sept.
  2019.

\bibitem{LF22}
H.~{Lou} and F.~{Farnoud}, ``Data deduplication with random substitutions,''
  \emph{IEEE Transactions on Information Theory}, vol.~68, no.~10, pp.
  6941--6963, Oct. 2022.

\bibitem{Zhang23}
Y.~{Zhang}, F.~{Zhang}, H.~{Li}, S.~{Zhang}, and X.~{Du}, ``Compressstreamdb:
  Fine-grained adaptive stream processing without decompression,'' in
  \emph{2023 IEEE 39th International Conference on Data Engineering
  (ICDE)}.\hskip 1em plus 0.5em minus 0.4em\relax Anaheim, CA, USA: IEEE, 2023,
  pp. 408--422.

\bibitem{Zhang24}
Y.~{Zhang}, F.~{Zhang}, H.~{Li}, S.~{Zhang}, X.~{Guo}, Y.~{Chen}, A.~{Pan}, and
  X.~{Du}, ``Data-aware adaptive compression for stream processing,''
  \emph{IEEE Transactions on Knowledge and Data Engineering}, vol.~36, no.~9,
  pp. 4531--4549, Sept. 2024.

\bibitem{DNA10}
K.~Daily, P.~Rigor, S.~Christley, X.~Xie, and P.~Baldi, ``Data structures and
  compression algorithms for high-throughput sequencing technologies,''
  \emph{{BMC} Bioinform.}, vol.~11, p. 514, Oct. 2010.

\bibitem{DNA13}
J.~J. Selva and X.~Chen, ``{SRC}omp: Short read sequence compression using
  burstsort and {Elias} omega coding,'' \emph{PLOS ONE}, vol.~8, no.~12, pp.
  1--7, 12 Dec. 2013.

\bibitem{NIPS}
D.~Alistarh, D.~Grubic, J.~Li, R.~Tomioka, and M.~Vojnovic, ``Compressstreamdb:
  Fine-grained adaptive stream processing without decompression,'' in
  \emph{31st Conference on Neural Information Processing Systems (NIPS
  2017)}.\hskip 1em plus 0.5em minus 0.4em\relax Long Beach, CA, USA: MIT
  Press, 2017, pp. 1709--1720.

\bibitem{YLH23}
W.~Yan, S.-J. Lin, and Y.~S. Han, ``A new metric and the construction for
  evolving 2-threshold secret sharing schemes based on prefix coding of
  integers,'' \emph{{IEEE} Transactions on Communications}, vol.~71, no.~5, pp.
  2906--2915, May 2023.

\bibitem{Cheng25}
Q.~{Cheng}, H.~{Cao}, S.-J. {Lin}, N.~{Yu}, Y.~S. {Han}, and X.~{Xie}, ``A
  construction of evolving k-threshold secret sharing scheme over a polynomial
  ring,'' in \emph{Advances in Cryptology -- ASIACRYPT 2025}.\hskip 1em plus
  0.5em minus 0.4em\relax Singapore: Nature Singapore, 2026, pp. 3--33.

\bibitem{C1990}
R.~M. Capocelli, ``Flag encodings related to the zeckendorf representation of
  integers,'' in \emph{Sequences, Combinatorics, Compression, Security, and
  Transmission}.\hskip 1em plus 0.5em minus 0.4em\relax New York, NY, USA:
  Springer-Verlag, 1990, pp. 449--466.

\bibitem{L68}
V.~I. Levenshtein, ``On the redundancy and delay of decodable coding of natural
  numbers (in {Russian}),'' \emph{Problems of Cybernetics}, vol.~20, pp.
  173--179, 1968.

\bibitem{ER78}
S.~{Even} and M.~{Rodeh}, ``Economical encoding of commas between strings,''
  \emph{Communications of the ACM}, vol.~21, no.~4, pp. 315--317, Apr. 1978.

\bibitem{S80}
Q.~F. Stout, ``Improved prefix encodings of the natural numbers (corresp.),''
  \emph{IEEE Transactions on Information Theory}, vol.~26, no.~5, pp. 607--609,
  Sep. 1980.

\bibitem{Y00}
H.~Yamamoto, ``A new recursive universal code of the positive integers,''
  \emph{IEEE Transactions on Information Theory}, vol.~46, no.~2, pp. 717--723,
  Mar. 2000.

\bibitem{AF87}
A.~Apostolico and A.~S. Fraenkel, ``Robust transmission of unbounded strings
  using {Fibonacci} representations,'' \emph{IEEE Transactions on Information
  Theory}, vol.~33, no.~2, pp. 238--245, Mar. 1987.

\bibitem{L81}
K.~{Lakshmanan}, ``On universal codeword sets,'' \emph{IEEE Transactions on
  Information Theory}, vol.~27, no.~5, pp. 659--662, Sep. 1981.

\bibitem{W88}
M.~Wang, ``Almost asymptotically optimal flag encoding of the integers,''
  \emph{IEEE Transactions on Information Theory}, vol.~34, no.~2, pp. 324--326,
  Mar. 1988.

\bibitem{YO91}
H.~Yamamoto and H.~Ochi, ``A new asymptotically optimal code for the positive
  integers,'' \emph{IEEE Transactions on Information Theory}, vol.~37, no.~5,
  pp. 1420--1429, Sep. 1991.

\bibitem{AS17}
B.~T. {\'{A}}vila and R.~M.~C. de~Souza, ``{Meta-Fibonacci} codes: Efficient
  universal coding of natural numbers,'' \emph{IEEE Transactions on Information
  Theory}, vol.~63, no.~4, pp. 2357--2375, Apr. 2017.

\bibitem{DCC21}
L.~Allison, A.~S. Konagurthu, and D.~F. Schmidt, ``On universal codes for
  integers: Wallace tree, {Elias} omega and beyond,'' in \emph{Proc. 2021 Data
  Compression Conference (DCC)}, Mar. 2021, pp. 313--322.

\bibitem{YL21}
W.~Yan and S.-J. Lin, ``On the minimum of the expansion factor for universal
  coding of integers,'' \emph{IEEE Transactions on Communications}, vol.~69,
  no.~11, pp. 7309--7319, Nov. 2021.

\bibitem{YL22}
W.~Yan and S.-J. Lin, ``A tighter upper bound of the expansion factor for
  universal coding of integers and its code constructions,'' \emph{IEEE
  Transactions on Communications}, vol.~70, no.~7, pp. 4429--4438, Jul. 2022.

\bibitem{ITW}
W.~Yan and S.-J. Lin, ``Generalized universal coding of integers,'' in
  \emph{Proc. IEEE Inf. Theory Workshop (ITW)}.\hskip 1em plus 0.5em minus
  0.4em\relax Kanazawa, Japan: IEEE, 2021, pp. 1--6.

\bibitem{YH24}
W.~Yan and Y.~S. Han, ``Generalized universal coding of integers,'' \emph{IEEE
  Transactions on Communications}, vol.~72, no.~8, pp. 4538--4550, Aug. 2024.

\bibitem{Yan25}
W.~{Yan}, Y.~S. {Han}, and G.~{Yang}, ``On some properties for universal coding
  of integers and its generalization,'' \emph{IEEE Transactions on
  Communications}, vol.~73, no.~10, pp. 8587--8595, Oct. 2025.

\bibitem{book07}
D.~Salomon, \emph{Variable-length Codes for Data Compression}.\hskip 1em plus
  0.5em minus 0.4em\relax London, U.K.: Springer-Verlag, 2007.

\bibitem{A1993}
T.~{Amemiya} and H.~{Yamamoto}, ``A new class of the universal representation
  for the positive integers,'' \emph{IEICE Transactions on Fundamentals of
  Electronics, Communications and Computer Sciences}, vol. E76A, no.~3, pp.
  447--452, Mar. 1993.

\bibitem{Abel}
P.~K. {Hung}, \emph{Secrets in Inequalities (volume 1)}.\hskip 1em plus 0.5em
  minus 0.4em\relax GIL Publishing House, 2007.

\bibitem{EIT}
T.~M. {Cover} and J.~A. {Thomas}, \emph{Elements of information theory, {2nd
  ed}}.\hskip 1em plus 0.5em minus 0.4em\relax NY, USA: Wiley, 2006.

\bibitem{kraft}
L.~G. Kraft, ``A device for quantizing, grouping, and coding
  amplitude-modulated pulses,'' Master's thesis, Dept. of Electrical
  Engineering, Massachusetts Institute of Technology, Cambridge, Mass., 1949.

\bibitem{Wyner72}
A.~D. Wyner, ``An upper bound on the entropy series,'' \emph{Inf. Control.},
  vol.~20, no.~2, pp. 176--181, Mar. 1972.

\bibitem{91NN}
N.~Nakatsu, ``Bounds on the redundancy of binary alphabetical codes,''
  \emph{IEEE Transactions on Information Theory}, vol.~37, no.~4, pp.
  1225--1229, July 1991.

\bibitem{91Yeung}
R.~Yeung, ``Alphabetic codes revisited,'' \emph{IEEE Transactions on
  Information Theory}, vol.~37, no.~3, pp. 564--572, May 1991.

\bibitem{2024code}
R.~Bruno, R.~De~Prisco, A.~De~Santis, and U.~Vaccaro, ``Bounds and algorithms
  for alphabetic codes and binary search trees,'' \emph{IEEE Transactions on
  Information Theory}, vol.~70, no.~10, pp. 6974--6988, Oct. 2024,.

\end{thebibliography}
\end{document}